\documentclass[review]{elsarticle}

\usepackage{lineno,hyperref}
\usepackage{amsmath}
\usepackage{amssymb}
\usepackage{algorithmic}
\usepackage{algorithm}
\usepackage{caption}
\usepackage{mathtools}
\usepackage{changes}
\usepackage{multirow}
\usepackage{verbatim}
\usepackage{mathrsfs}
\usepackage{graphicx}
\usepackage{amsmath,amsthm,bm,mathrsfs}

\theoremstyle{plain}
\newtheorem{theorem}{Theorem}[section]
\newtheorem{lemma}[theorem]{Lemma}

\newtheorem{proposition}[theorem]{Proposition}

\theoremstyle{definition}
\newtheorem{definition}[theorem]{Definition}

\theoremstyle{remark}
\newtheorem{remark}{Remark}

\modulolinenumbers[5]

\journal{Journal of \LaTeX\ Templates}

\begin{document}

\begin{frontmatter}

\title{Frequency-limited $\mathcal{H}_2$ Model Order Reduction Based on Relative Error}

\author[mymainaddress]{Umair~Zulfiqar}

\author[mymainaddress,mysecondaryaddress]{Xin~Du\corref{mycorrespondingauthor}}
\cortext[mycorrespondingauthor]{Corresponding author}
\ead{duxin@shu.edu.cn}

\author[mymainaddress]{Qiuyan~Song}

\author[ml]{Zhi-Hua~Xiao}

\author[vs]{Victor~Sreeram}

\address[mymainaddress]{School of Mechatronic Engineering and Automation, Shanghai University, Shanghai, 200444, China}
\address[mysecondaryaddress]{Shanghai Key Laboratory of Power Station Automation Technology, Shanghai University, Shanghai, 200444, China}
\address[ml]{School of Information and Mathematics, Yangtze University, Jingzhou, Hubei, 434023, China}
\address[vs]{Department of Electrical, Electronic, and Computer Engineering, The University of Western Australia, Perth, 6009, Australia}

\begin{abstract}
Frequency-limited model order reduction aims to approximate a high-order model with a reduced-order model that maintains high fidelity within a specific frequency range. Beyond this range, a decrease in accuracy is acceptable due to the nature of the problem. The quality of the reduced-order model is typically evaluated using absolute or relative measures of approximation error. Relative error, which represents the percentage error, becomes particularly relevant when reducing a plant model for the purpose of designing a reduced-order controller. This paper derives the necessary conditions for achieving a local optimum of the frequency-limited $\mathcal{H}_2$ norm for the relative error system. Based on these optimality conditions, an oblique projection algorithm is proposed to ensure a small relative error within the desired frequency interval. Unlike existing algorithms, the proposed approach does not necessitate solving large-scale Lyapunov and Ricatti equations. Instead, the proposed algorithm relies on solving sparse-dense Sylvester equations, which typically emerge in the majority of $\mathcal{H}_2$ model order reduction algorithms, but can be efficiently solved. To evaluate the performance of the proposed algorithm, a comparison is conducted with three existing techniques: frequency-limited balanced truncation, frequency-limited balanced stochastic truncation, and frequency-limited iterative Rational Krylov algorithm. The comparative analysis focuses on designing reduced-order controllers for high-order plants. Numerical results confirm that the reduced-order controllers obtained using the proposed algorithm ensure superior robust closed-loop stability.
\end{abstract}


\end{frontmatter}

\section{Introduction}
In the last few decades, computer-aided simulation has emerged as a key catalyst for facilitating rapid scientific discoveries and technological advancements. By utilizing computer programs, mathematical models representing both natural and artificial dynamical systems can be simulated, enabling analysis across a vast array of scenarios and configurations. Typically, the dynamical behavior of these systems is captured through mathematical models formulated as partial differential equations (PDEs). To harness the rich mathematical tools available in dynamical system theory and achieve a comprehensive understanding of the system's behavior, these PDEs are converted into ordinary differential equations (ODEs) \cite{benner2017model,quarteroni2014reduced,antoulas2020interpolatory}.

Consider a state-space model of a dynamical system $H$, where the inputs, outputs, and states are represented by $u(t)\in\mathbb{R}^{m\times 1}$, $y(t)\in\mathbb{R}^{p\times 1}$, and $x(t)\in\mathbb{R}^{n\times 1}$, respectively. The state-space equations describing the system $H$ are given by:
\begin{align}
\dot{x}(t)&=Ax(t)+Bu(t),\nonumber\\
y(t)&=Cx(t)+Du(t),\nonumber
\end{align}
where the dimensions of the matrices in the above equations are as follows: $A\in\mathbb{R}^{n\times n}$, $B\in\mathbb{R}^{n\times m}$, $C\in\mathbb{R}^{p\times n}$, and $D\in\mathbb{R}^{p\times m}$. The $n^{th}$-order $p\times m$ transfer function $H(s)$ of $H$ can be related to its state-space representation $(A,B,C,D)$ as follows:
\begin{align}
H(s)=C(sI-A)^{-1}B+D.\nonumber
\end{align}

In order to simulate $H$, the solution of $n$ ordinary differential equations (ODEs) is required for each specific scenario. The continuous advancements in chip manufacturing technology have led to significant improvements in the speed, memory capacity, and processing power of modern computers \cite{benner2019system}. However, these technological advancements have also resulted in the increasing complexity of modern dynamical systems, leading to practical mathematical models with dimensions reaching several thousand. As a result, efficiently simulating and analyzing these high-order models remains a computational challenge. This challenge has motivated the development of computationally efficient model order reduction (MOR) algorithms, aiming to reduce the computational cost of simulations \cite{benner2005dimension}. MOR involves constructing a reduced-order model (ROM) that provides an accurate approximation of the original high-order model. The ROM serves as a surrogate for the original model during simulation and analysis, exhibiting similar behavior and characteristics.

Let us consider the ROM $\tilde{H}$, which is characterized by the states $\tilde{x}(t)$ and outputs $\tilde{y}(t)$. The dynamics of $\tilde{H}$ are described by the following state-space equations:
\begin{align}
\dot{\tilde{x}}(t)&=\tilde{A}\tilde{x}(t)+\tilde{B}u(t),\nonumber\\
\tilde{y}(t)&=\tilde{C}\tilde{x}(t)+\tilde{D}u(t).\nonumber
\end{align}
The dimensions of the matrices involved in these equations are as follows: $\tilde{A}\in\mathbb{R}^{r\times r}$, $\tilde{B}\in\mathbb{R}^{r\times m}$, $\tilde{C}\in\mathbb{R}^{p\times r}$, and $\tilde{D}\in\mathbb{R}^{p\times m}$. The $r^{th}$-order $p\times m$ transfer function $\tilde{H}(s)$ of $\tilde{H}$ can be expressed in terms of its state-space realization $(\tilde{A},\tilde{B},\tilde{C},\tilde{D})$ as follows:
\begin{align}
\tilde{H}(s)=\tilde{C}(sI-\tilde{A})^{-1}\tilde{B}+\tilde{D}.\nonumber
\end{align}
The primary objective of model order reduction (MOR) is to ensure that the ROM $\tilde{H}$ closely emulates the behavior of the original system $H$, while preserving crucial properties such as stability, passivity, contractivity, and more. Additionally, it is desirable to minimize the computational cost required to obtain $\tilde{H}$ as well as the order $r$ of the reduced model.

The MOR problem can be mathematically formulated as follows:
\begin{enumerate}
\item $H=\tilde{H}+\Delta_a$.
\item $H=\tilde{H}(I+\Delta_r)$.
\end{enumerate}
The objective is to minimize $\Delta_a$ and $\Delta_r$ in order to reduce the approximation error. The quality of the approximation is often measured using the $\mathcal{H}_2$ and $\mathcal{H}_\infty$ norms, which are commonly used system norms. Existing algorithms that employ $\Delta_r$ as the approximation criterion tend to be computationally expensive, making $\Delta_r$ less frequently used compared to $\Delta_a$. However, choosing the relative error as a reduction criterion is motivated by the objective of designing a low-order controller for the reduced-order plant. In this context, $||\Delta_r(s)||$ serves as a more meaningful and theoretically sound criterion to ensure closed-loop stability when connecting the controller to the original high-order plant. To achieve this, the MOR algorithm employed for constructing $\tilde{H}(s)$ should minimize the following reduction criterion:
\begin{align}
\underset{\substack{\tilde{H}(s)\\\textnormal{order}=r}}{\text{min}}||[I+K(s)\tilde{H}(s)]^{-1}K(s)\tilde{H}(s)\Delta_r(s)||,
\end{align}
where $\Delta_r(s)=\tilde{H}^{-1}(s)\Delta_{a}(s)$. Alternatively, the reduction criterion can be expressed as:
\begin{align}
\underset{\substack{\tilde{H}(s)\\\textnormal{order}=r}}{\text{min}}||[\Delta_r(s)]\tilde{H}(s)K(s)[I+\tilde{H}(s)K(s)]^{-1}||,
\end{align}
with $\Delta_r(s)=\Delta_a(s)\tilde{H}^{-1}(s)$. In these reduction criteria, $K(s)$ and $\tilde{H}(s)$ are unknown since $K(s)$ is designed for $\tilde{H}(s)$ after reducing $H(s)$. Assuming a well-designed control system, it is reasonable to expect that $\tilde{H}(s)K(s)[I+\tilde{H}(s)K(s)]^{-1}$ or $[I+K(s)\tilde{H}(s)]^{-1}K(s)\tilde{H}(s)$ approximates $I$ within the operating bandwidth of the system and approximates $0$ outside the bandwidth. The numerical results section will highlight that a robust control design for the reduced-order plant can satisfy this condition. Under this assumption, the problem at hand can be viewed as a frequency-limited relative error MOR problem. If $||\Delta_r(s)||$ is small within the desired frequency interval (operating bandwidth), connecting $K(s)$ to $H(s)$ will ensure good closed-loop stability. Any error outside this desired frequency range is attenuated by $\tilde{H}(s)K(s)[I+\tilde{H}(s)K(s)]^{-1}$ or $[I+K(s)\tilde{H}(s)]^{-1}K(s)\tilde{H}(s)$ due to their approximate frequency response of $0$. Furthermore, the task of obtaining a reduced-order infinite impulse response (IIR) filter from a high-order finite impulse response (FIR) filter can also be seen as a frequency-limited relative error MOR problem.

Various families of MOR algorithms have been developed based on different approximation criteria, preservation requirements, and methodological approaches to achieve the desired objectives. This paper provides a brief survey of important algorithms in the literature  that are relevant to the topic at hand. Among these algorithms, balanced truncation (BT) is one of the most important and renowned MOR algorithms, known for its advantageous features such as stability preservation, \textit{apriori} error bound, and accuracy \cite{moore1981principal}. Through extensions and modifications documented in the literature, BT has evolved into a diverse family of algorithms known as balancing-based MOR algorithms. While BT is suitable for medium-scale systems, its applicability is extended to large-scale systems by utilizing low-rank and cost-effective solutions to Lyapunov equations, as discussed in \cite{kurschner2020inexact,mehrmann2005balanced}. The BT family encompasses algorithms that also preserve other crucial properties, such as passivity and contractivity, as explored in \cite{phillips2003guaranteed,yan2008second,phillips2004poor,reis2010positive}. Among the algorithms in the BT family, the most relevant to the problem addressed in this paper is the balanced stochastic truncation (BST) \cite{green1988balanced}. BST ensures a small relative error with an upper bound on $||\Delta_r(s)||_{\mathcal{H}_\infty}$, as established in \cite{green1988relative}. Additionally, BST preserves the stable minimum phase characteristic of the original model.

In the existing literature, most MOR algorithms that utilize the $\mathcal{H}_2$ norm of error as their approximation criterion (in contrast to those using the $\mathcal{H}_\infty$ norm of error) are computationally efficient and suitable for large-scale systems \cite{wolf2014h}. Wilson's conditions, derived in \cite{wilson1970optimum}, provide the necessary conditions for a local optimum of $||\Delta_a(s)||_{\mathcal{H}_2}^2$. These conditions, shown to be equivalent to specific interpolation conditions in \cite{gugercin2008h_2,van2008h2}, indicate that a local optimal $\tilde{H}(s)$ interpolates $H(s)$ at selected points in the $s$-plane. The pioneering interpolation algorithm for constructing a local optimum of $||\Delta_a(s)||_{\mathcal{H}_2}^2$ in single-input single-output (SISO) systems is the iterative rational Krylov algorithm (IRKA) \cite{gugercin2008h_2}. The applicability of IRKA has been extended to multi-input multi-output (MIMO) systems in \cite{van2008h2}. Another algorithm based on Sylvester equations and projection, which exploits their connection, is presented in \cite{xu2011optimal} and can construct a local optimum of $||\Delta_{a}(s)||_{\mathcal{H}_2}^2$. Some of the numerical weaknesses of this algorithm have been addressed in \cite{benner2011sparse} to enhance its robustness. The convergence of IRKA slows down as the number of inputs and outputs increases. To mitigate this, trust region-based methods have been proposed to expedite the convergence of IRKA, such as those discussed in \cite{beattie2009trust}. Moreover, a family of algorithms based on manifold theory has been reported in the literature, capable of constructing a local optimum of $||\Delta_a(s)||_{\mathcal{H}_2}^2$, as explored in \cite{jiang2017h2,sato2017structure,sato2015riemannian,yang2019trust}. Additionally, there are suboptimal algorithms with additional properties, such as stability preservation, as reported in \cite{panzer2013greedy,wolf2013h,ibrir2018projection}.

In the frequency-limited MOR problem, the primary objective is to achieve superior accuracy within a specific desired frequency interval, rather than aiming for overall accuracy across the entire frequency range. Errors occurring outside this interval are considered tolerable. Many MOR problems inherently exhibit frequency-limited nature because certain frequency intervals are of greater importance. For example, when creating a reduced-order model (ROM) for a notch filter, minimizing the approximation error near the notch frequency becomes crucial. To ensure closed-loop stability, the ROM of the plant must accurately capture the system's behavior in the crossover frequency region \cite{enns1984model}. In interconnected power systems, the presence of low-frequency oscillations is critical for small-signal stability studies. Thus, the ROM of interconnected power systems should accurately represent the behavior within the frequency range that encompasses inter-area and inter-plant oscillations \cite{zulfiqar2021h2}. Gawronski et al. \cite{gawronski1990model} extended the balanced truncation (BT) method to address the frequency-limited MOR problem. The applicability of FLBT \cite{gawronski1990model} is primarily limited to medium-scale systems. To extend FLBT to large-scale systems, Benner et al. \cite{benner2016frequency} introduced low-rank solutions of Lyapunov equations. Furthermore, the relative-error case of the frequency-limited MOR problem is explored in \cite{shaker2008frequency}, where the balanced stochastic truncation (BST) method is generalized to obtain frequency-limited BST (FLBST). However, the applicability of FLBST is mainly confined to small-scale systems due to the requirement of dense Lyapunov and Ricatti equation solutions.

The definition of the frequency-limited $\mathcal{H}_{2}$ norm ($\mathcal{H}_{2,\omega}$ norm) is introduced in \cite{petersson2014model}, along with necessary conditions based on gramians for achieving a local optimum of $||\Delta_a(s)||_{\mathcal{H}_{2,\omega}}^2$. Generalizations of the iterative rational Krylov algorithm (IRKA) are proposed in \cite{vuillemin2013h2,zulfiqar2021frequency}, resulting in a frequency-limited IRKA (FLIRKA) algorithm that approximately satisfies these necessary conditions. The interpolation-based optimality conditions are derived in \cite{vuillemin2014frequency}, and their equivalence to the gramian-based conditions is established in \cite{zulfiqar2021h2,zulfiqar2019adaptive}. A specific interpolation algorithm that fulfills a subset of these optimality conditions while ensuring stability is put forth in \cite{zulfiqar2019adaptive}. Notably, there is currently no algorithm in the literature that specifically addresses the minimization of $||\Delta_r(s)||_{\mathcal{H}_{2,\omega}}$. This research gap serves as the motivation for the present study.

This paper focuses on deriving necessary conditions, based on gramians, for achieving the local optimum of $||\Delta_r(s)||_{\mathcal{H}_{2,\omega}}^2$. It is demonstrated that constructing a local optimum of $||\Delta_{r}(s)||_{\mathcal{H}_{2,\omega}}^2$ within an oblique projection framework is not possible. However, the paper proposes an efficient oblique projection algorithm that specifically targets the local optimum of $||\Delta_{r}(s)||_{\mathcal{H}_{2,\omega}}^2$, offering high fidelity. Unlike FLBST, the proposed algorithm only requires solving small-scale Lyapunov and Ricatti equations, even in large-scale scenarios. The main computational burden lies in solving sparse-dense Sylvester equations, which are commonly encountered in most $\mathcal{H}_2$-optimal MOR algorithms. Unlike the large-scale Lyapunov and Ricatti equations, these sparse-dense Sylvester equations can be solved efficiently. The proposed algorithm is used to design reduced-order controllers for three benchmark high-order systems, widely utilized in the literature for evaluating MOR algorithms. Comparisons are made between the accuracy of the proposed algorithm and that of FLBT, FLBST, and FLIRKA. Additionally, the robust stability of the reduced-order controllers developed for the ROMs generated by these algorithms is examined. Numerical findings consistently demonstrate the superiority of the proposed algorithm over existing algorithms.
\section{Preliminaries}
Consider an $n$-th order $m\times m$ stable square linear time-invariant system denoted by $H$. Within the desired frequency interval $[0,\Omega]$ rad/sec, the frequency-limited controllability gramian $P$ and the frequency-limited observability gramian $Q$ of the state-space realization $(A,B,C)$ can be expressed as integral equations, i.e.,
\begin{align}
P=\frac{1}{2\pi}\int_{-\Omega}^{\Omega}(j\nu I-A)^{-1}BB^T(j\nu I-A)^{-*}d\nu,\nonumber\\
Q=\frac{1}{2\pi}\int_{-\Omega}^{\Omega}(j\nu I-A)^{-*}C^TC(j\nu I-A)^{-1}d\nu.\nonumber
\end{align}
We can define $S_\omega[A]$ as $\frac{j}{2\pi}\ln\big((A+j\Omega I)(A-j\Omega I)^{-1}\big)$, where $(\cdot)^*$ represents Hermitian. By utilizing this definition, the controllability gramian $P$ and the observability gramian $Q$ satisfy the Lyapunov equations expressed as
\begin{align}
AP+PA^T+S_\omega[A]BB^T+BB^TS_\omega[A^T]&=0,\nonumber\\
A^TQ+QA+S_\omega[A^T]C^TC+C^TCS_\omega[A]&=0,\nonumber
\end{align} respectively. For more information, please refer to \cite{gawronski1990model,petersson2014model}.
\begin{definition}
The frequency-limited $\mathcal{H}_2$ norm of $H(s)$ within the desired frequency interval $[0,\Omega]$ rad/sec, as defined in \cite{petersson2014model,vuillemin2014frequency}, can be expressed as
\begin{align}
||H(s)||_{\mathcal{H}_{2,\omega}}&=\sqrt{\frac{1}{2\pi}trace\Big(\int_{-\Omega}^{\Omega}H(j\nu)H^{*}(j\nu)d\nu\Big)},\nonumber\\
&=\sqrt{trace\Big(CPC^T+2\big(CS_\omega[A]B+D\frac{\Omega}{2\pi}\big)D^T\Big)},\nonumber\\
&=\sqrt{\frac{1}{2\pi}trace\Big(\int_{-\Omega}^{\Omega}H^*(j\nu)H(j\nu)d\nu}\Big),\nonumber\\
&=\sqrt{trace\Big(B^TQB+2\big(CS_\omega[A]B+D\frac{\Omega}{2\pi}\big)D^T\Big)}.\nonumber
\end{align}
\end{definition}
\subsection{Problem Formulation}
The state-space matrices of the ROM within an oblique projection framework can be obtained from the original model using the following computations:
\begin{align}
\tilde{A}=\tilde{W}^TA\tilde{V},\quad \tilde{B}=\tilde{W}^TB,\quad \tilde{C}=C\tilde{V},\quad \tilde{D}=D.\nonumber
\end{align}
In these equations, $\tilde{V}\in\mathbb{R}^{n\times r}$, $\tilde{W}\in\mathbb{R}^{n\times r}$, $\tilde{W}^T\tilde{V}=I$, and the columns of $\tilde{V}$ form an $r$-dimensional subspace along with the kernels of $\tilde{W}^T$. Moreover, the matrix $\Pi=\tilde{V}\tilde{W}^T$ serves as an oblique projector \cite{antoulas2005approximation}.

In the context of oblique projection, the primary objective of the relative error frequency-limited $\mathcal{H}_2$ MOR problem is to compute the reduction matrices $\tilde{V}\in\mathbb{R}^{n\times r}$ and $\tilde{W}\in\mathbb{R}^{n\times r}$. These matrices are chosen to ensure that the ROM $\tilde{H}(s)$ yields a small value of $||\Delta_{r}(s)||_{\mathcal{H}_{2,\omega}}$. Essentially, the problem aims to minimize $||\Delta_{r}(s)||_{\mathcal{H}_{2,\omega}}$ for a ROM $\tilde{H}(s)$ with an order of $r$, i.e.,
\begin{align}
\underset{\substack{\tilde{H}(s)\\\textnormal{order}=r}}{\text{min}}||\Delta_{r}(s)||_{\mathcal{H}_{2,\omega}}.\nonumber
\end{align}
\subsection{Existing Oblique Projection Algorithms}
The following subsection offers a brief summary and examination of three oblique projection algorithms that are specifically related to frequency-limited MOR. These algorithms have direct relevance to the problem at hand.
\subsubsection{Frequency-limited Balanced Truncation (FLBT)}
In FLBT \cite{gawronski1990model}, the computation of reduction matrices involves applying the contragradient transformation to $P$ and $Q$. This computation yields $\tilde{W}^TP\tilde{W}=\tilde{V}^TQ\tilde{V}= diag(\sigma_1,\cdots,\sigma_r)$, where $\sigma_{i+1}\geq\sigma_i$. As a result, the ROM obtained through this process ensures that $\Delta_{a}(j\omega)$ remains small within the desired frequency interval of $[0,\Omega]$ rad/sec.
\subsubsection{Frequency-limited Balanced Stochastic Truncation (FLBST)}
The controllability gramian $P_c$ of the pair $(A,B)$ solves the Lyapunov equation
\begin{align}
AP_c+P_cA^T+BB^T=0.\nonumber
\end{align} By defining $B_f$ and $A_f$ as $B_f=P_cC^T+BD^T$ and $A_f=A-B_f(DD^T)^{-1}C$, respectively, we can obtain $X_f$ by solving the Riccati equation
\begin{align}
A_f^TX_f+X_f A_f+X_f B_f(DD^T)^{-1}B_f^TX_f+C^T(DD^T)^{-1}C=0.\nonumber
\end{align}
 The frequency-limited observability gramian $Q_\Omega$ of the pair $\big(A,D^{-1}(C-B_f^TX_f)\big)$ satisfies the Lyapunov equation
\begin{align}
A^TQ_\Omega+Q_\Omega A&+S_\omega[A^T]\big(D^{-1}(C-B_f^TX_f)\big)^T\big(D^{-1}(C-B_f^TX_f)\big)\nonumber\\
&+\big(D^{-1}(C-B_f^TX_f)\big)^T\big(D^{-1}(C-B_f^TX_f)\big)S_\omega[A]=0.\nonumber
\end{align}
The reduction matrices in FLBST \cite{shaker2008frequency} are obtained through the contragradient transformation of $P$ and $Q_\Omega$ as $\tilde{W}^TP\tilde{W}=\tilde{V}^TQ_\Omega\tilde{V}= diag(\sigma_1,\cdots,\sigma_r)$, where $\sigma_{i+1}\geq\sigma_i$. Consequently, the resulting ROM guarantees small values of $\Delta_{r}(j\omega)$ within the desired frequency interval of $[0,\Omega]$ rad/sec.
\subsubsection{Frequency-limited Iterative Rational Krylov Algorithm (FLIRKA)}
Let $P_{12}$ and $Y_\Omega$ satisfy the following Sylvester equations
\begin{align}
AP_{12}+P_{12}\tilde{A}^T+S_\omega[A]B\tilde{B}^T+B\tilde{B}^TS_\omega[\tilde{A}^T]&=0,\label{e6}\\
A^TY_\Omega+Y_\Omega\tilde{A}-S_\omega[A^T]C^T\tilde{C}-C^T\tilde{C}S_\omega[\tilde{A}]&=0.
\end{align} The reduction matrices in FLIRKA \cite{vuillemin2013h2} are iteratively updated starting from an arbitrary guess $(\tilde{A},\tilde{B},\tilde{C})$. Specifically, $\tilde{V}$ is set to $P_{12}$ and $\tilde{W}$ is set to $Y_\Omega(P_{12}^TY_\Omega)^{-1}$. Upon convergence, a ROM is obtained, ensuring small values of $||\Delta_a(s)||_{\mathcal{H}_{2,\omega}}$.
\section{Main Results}
This section focuses on formulating a state-space representation for $\Delta_r$ and presenting an expression for $||\Delta_r(s)||_{\mathcal{H}_{2,\omega}}$. Subsequently, we derive the necessary conditions that must be satisfied to obtain a local optimum for $||\Delta_r(s)||_{\mathcal{H}_{2,\omega}}^2$. Based on these optimality conditions, we introduce an oblique projection algorithm aimed at constructing a local optimum.
\subsection{State-space Realization and $\mathcal{H}_{2,\omega}$ norm of $\Delta_r$}
Consider a stable minimum phase system represented by $\tilde{H}(s)$. In this case, $\Delta_r(s)$ can be written as $\Delta_r(s) = \tilde{H}^{-1}(s)\Delta_a(s)$. When matrix $D$ has full rank, a state-space realization $\tilde{H}^{-1}(s)$ can be obtained using the state-space realization $(\tilde{A}, \tilde{B}, \tilde{C}, D)$ of $\tilde{H}(s)$. This realization is given by:
\begin{align}
A_i&=\tilde{A}-\tilde{B}D^{-1}\tilde{C},& B_i&=-\tilde{B}D^{-1},& C_i&=D^{-1}\tilde{C},& D_i&=D^{-1},\nonumber
\end{align}cf. \cite{zhou1995frequency,zhou1996robust}. Furthermore, a state-space realization of $\tilde{H}^{-1}(s)\Delta_a(s)$ can be expressed as:
\begin{align}
A_r&=\begin{bmatrix}A&0&0\\0&\tilde{A}&0\\B_iC& -B_i\tilde{C}&A_i\end{bmatrix},&B_r&=\begin{bmatrix}B\\\tilde{B}\\0\end{bmatrix},\nonumber\\
C_r&=\begin{bmatrix}D_iC&-D_i\tilde{C}&C_i\end{bmatrix}.\nonumber
\end{align}
The triangular structure of the matrix $A_r$ allows us to express the matrix logarithm $S_\omega[A_r]$ as:
\begin{align}
S_\omega[A_{r}]=\begin{bmatrix}S_\omega[A]&0&0\\0& S_\omega[\tilde{A}]&0\\S_{\omega,1} &S_{\omega,2} &S_\omega[A_i]\end{bmatrix},\nonumber
\end{align}cf. \cite{cheng2001approximating}. Since $A_rS_\omega[A_r] = S_\omega[A_r]A_r$, cf. \cite{higham2008functions}, we have the following relationship:
\begin{align}
\begin{bsmallmatrix}\star&&&&\star&&\star\\\star&&&&\star&&\star\\A_iS_{\omega,1}-S_{\omega,1}A+B_iCS_\omega[A]-S_\omega[A_i]B_iC&&&&A_iS_{\omega,2}-S_{\omega,2}\tilde{A}-B_i\tilde{C}S_\omega[\tilde{A}]+S_\omega[A_i]B_i\tilde{C}&&\star\end{bsmallmatrix}=0.\nonumber
\end{align} Consequently, $S_{\omega,1}$ and $S_{\omega,2}$ satisfy the following Sylvester equations:
\begin{align}
A_iS_{\omega,1}-S_{\omega,1}A+B_iCS_\omega[A]-S_\omega[A_i]B_iC&=0,\label{e8}\\
A_iS_{\omega,2}-S_{\omega,2}\tilde{A}-B_i\tilde{C}S_\omega[\tilde{A}]+S_\omega[A_i]B_i\tilde{C}&=0.\label{e9}
\end{align}
\begin{lemma}
The matrix $S_{\omega,2}$ is equal to $S_\omega[\tilde{A}] - S_\omega[A_i]$.
\end{lemma}
\begin{proof}
We start by considering the equation:
\begin{align}
A_iS_{\omega,2}-S_{\omega,2}\tilde{A}&=B_i\tilde{C}S_\omega[\tilde{A}]-S_\omega[A_i]B_i\tilde{C}.\nonumber
\end{align}
Since $A_iS_\omega[A_i] - S_\omega[A_i]A_i = 0$, we observe that:
\begin{align}
-S_\omega[A_i]B_i\tilde{C}=-B_i\tilde{C}S_\omega[A_i]+S_\omega[A_i]\tilde{A}-\tilde{A}S_\omega[A_i].\nonumber
\end{align}
Thus, we have:
\begin{align}
A_iS_{\omega,2}-S_{\omega,2}\tilde{A}&=B_i\tilde{C}S_\omega[\tilde{A}]-B_i\tilde{C}S_\omega[A_i]+S_\omega[A_i]\tilde{A}-\tilde{A}S_\omega[A_i]\nonumber\\
&=B_i\tilde{C}S_\omega[\tilde{A}]-B_i\tilde{C}S_\omega[A_i]+S_\omega[A_i]\tilde{A}-\tilde{A}S_\omega[A_i]\nonumber\\
&\hspace*{2cm}+\tilde{A}S_\omega[\tilde{A}]-S_\omega[\tilde{A}]\tilde{A}\nonumber\\
&=A_i(S_\omega[\tilde{A}]-S_\omega[A_i])-(S_\omega[\tilde{A}]-S_\omega[A_i])\tilde{A}.\nonumber
\end{align}
By the uniqueness of the solution to the Sylvester equation (\ref{e9}), we conclude that $S_{\omega,2} = S_\omega[\tilde{A}] - S_\omega[A_i]$.
\end{proof}

Let $P_{r}=\begin{bmatrix}P&P_{12}&P_{13}\\P_{12}^T&P_{22}&P_{23}\\P_{13}^T&P_{23}^T&P_{33}\end{bmatrix}$ represent the frequency-limited controllability gramian of the pair $(A_{r},B_{r})$. This gramian satisfies the Lyapunov equation \eqref{e10}, given by
\begin{align}
A_{r}P_{r}+P_{r}A_{r}^T+S_\omega[A_{r}]B_{r}B_{r}^T+B_{r}B_{r}^TS_\omega[A_{r}^T]&=0.\label{e10}
\end{align} By expanding equation \eqref{e10}, we observe that $P_{13}$, $P_{22}$, $P_{23}$, and $P_{33}$ solve the equations \eqref{e11}, \eqref{e12}, \eqref{e13}, and \eqref{e14} respectively, with the corresponding expressions for $X_{13}$, $X_{22}$, $X_{23}$, and $X_{33}$, i.e.,
\begin{align}
AP_{13}+P_{13}A_i^T+X_{13}&=0,\label{e11}\\
\tilde{A}P_{22}+P_{22}\tilde{A}^T+X_{22}&=0,\label{e12}\\
\tilde{A}P_{23}+P_{23}A_i^T+X_{23}&=0,\label{e13}\\
A_iP_{33}+P_{33}A_i^T+X_{33}&=0,\label{e14}
\end{align}wherein
\begin{align}
X_{13}&=P_{11}C^TB_i^T-P_{12}\tilde{C}^TB_i^T+BB^TS_{\omega,1}^T+B\tilde{B}^TS_{\omega,2}^T,\nonumber\\
X_{22}&=S_\omega[\tilde{A}]\tilde{B}\tilde{B}^T+\tilde{B}\tilde{B}^TS_\omega[\tilde{A}^T],\nonumber\\
X_{23}&=P_{12}^TC^TB_i^T-P_{22}\tilde{C}^TB_i^T+\tilde{B}B^TS_{\omega,1}^T+\tilde{B}\tilde{B}^TS_{\omega,2}^T,\nonumber\\
X_{33}&=B_iCP_{13}-B_i\tilde{C}P_{23}+P_{13}^TC^TB_i^T-P_{23}^T\tilde{C}^TB_i^T.\nonumber
\end{align}
Let $Q_{r}=\begin{bmatrix}Q_{11}&Q_{12}&Q_{13}\\Q_{12}^T&Q_{22}&Q_{23}\\Q_{13}^T&Q_{23}&Q_{33}\end{bmatrix}$ represent the frequency-limited observability gramian of the pair $(A_{r},C_{r})$. This gramian satisfies the Lyapunov equation \eqref{e15}, expressed as
\begin{align}
A_{r}^TQ_{r}+Q_{r}A_{r}+S_\omega[A_{r}^T]C_{r}^TC_{r}+C_{r}^TC_{r}S_\omega[A_{r}]&=0.\label{e15}
\end{align}
By expanding equation \eqref{e15}, we find that $Q_{11}$, $Q_{12}$, $Q_{13}$, $Q_{22}$, $Q_{23}$, and $Q_{33}$ satisfy equations \ref{e16}, \ref{e17}, \ref{e18}, \ref{e19}, \ref{e20}, and \ref{e21} respectively, with the corresponding expressions for $Y_{11}$, $Y_{12}$, $Y_{13}$, $Y_{22}$, $Y_{23}$, and $Y_{33}$, i.e.,
\begin{align}
A^TQ_{11}+Q_{11}A+Y_{11}&=0,\label{e16}\\
A^TQ_{12}+Q_{12}\tilde{A}+Y_{12}&=0,\label{e17}\\
A^TQ_{13}+Q_{13}A_i+Y_{13}&=0,\label{e18}\\
\tilde{A}^TQ_{22}+Q_{22}\tilde{A}+Y_{22}&=0,\label{e19}\\
\tilde{A}^TQ_{23}+Q_{23}A_i+Y_{23}&=0,\label{e20}\\
A_i^TQ_{33}+Q_{33}A_i+Y_{33}&=0,\label{e21}
\end{align}wherein
\begin{align}
Y_{11}&=C^TB_i^TQ_{13}^T+Q_{13}B_iC+S_\omega[A^T]C^TD_i^TD_iC+C^TD_i^TD_iCS_\omega[A]\nonumber\\
&\hspace*{2.25cm}+S_{\omega,1}^TC_i^TD_iC+C^TD_i^TC_iS_{\omega,1},\nonumber\\
Y_{12}&=C^TB_i^TQ_{23}-Q_{13}B_i\tilde{C}-S_\omega[A^T]C^TD_i^TD_i\tilde{C}-C^TD_i^TD_i\tilde{C}S_\omega[\tilde{A}]\nonumber\\
&\hspace*{2.25cm}-S_{\omega,1}^TC_i^TD_i\tilde{C}+C^TD_i^TC_iS_{\omega,2},\nonumber\\
Y_{13}&=C^TB_i^TQ_{33}+S_\omega[A^T]C^TD_i^TC_i+C^TD_i^TC_iS_\omega[A_i]+S_{\omega,1}^TC_i^TC_i,\nonumber\\
Y_{22}&=-\tilde{C}^TB_i^TQ_{23}-Q_{23}B_i\tilde{C}+S_\omega[\tilde{A}^T]\tilde{C}^TD_i^TD_i\tilde{C}+\tilde{C}^TD_i^TD_i\tilde{C}S_\omega[\tilde{A}]\nonumber\\
&\hspace*{2.25cm}-S_{\omega,2}^TC_i^TD_i\tilde{C}-\tilde{C}^TD_i^TC_iS_{\omega,2},\nonumber\\
Y_{23}&=-\tilde{C}^TB_i^TQ_{33}-S_\omega[\tilde{A}^T]\tilde{C}^TD_i^TC_i-\tilde{C}^TD_i^TC_iS_\omega[A_i]+S_{\omega,2}^TC_i^TC_i,\nonumber\\
Y_{33}&=S_\omega[A_i^T]C_i^TC_i+C_i^TC_iS_\omega[A_i].\nonumber
\end{align}
\begin{proposition}
The matrix $Q_{12}$ is equal to $-Q_{13}$, and the matrix $Q_{22}$ is equal to $-Q_{23}$ and $Q_{33}$.
\end{proposition}
\begin{proof}
To demonstrate this, we can observe that by adding equations (\ref{e20}) and (\ref{e21}), we obtain
\begin{align}
A_i^T(Q_{23}+Q_{33})+(Q_{23}+Q_{33})A_i&=0.\nonumber
\end{align}
Consequently, it follows that $Q_{23}+Q_{33}=0$ and $Q_{33}=-Q_{23}$. Similarly, by subtracting equation (\ref{e19}) from (\ref{e21}), we find
\begin{align}
A_i^T(Q_{33}-Q_{22})+(Q_{33}-Q_{22})A_i&=0,\nonumber
\end{align}which yields $Q_{33}-Q_{22}=0$ and $Q_{33}=Q_{22}$. Lastly, adding equations (\ref{e16}) and (\ref{e17}) allows us to obtain
\begin{align}
A^T(Q_{12}+Q_{13})+(Q_{12}+Q_{13})\tilde{A}&=0,\nonumber
\end{align}implying $Q_{12}+Q_{13}=0$ and thus $Q_{12}=-Q_{13}$.
\end{proof}
The expression for $||\Delta_{r}(s)||_{\mathcal{H}_{2,\omega}}$ can be derived from the definition of the $\mathcal{H}_{2,\omega}$ norm and written as:
\begin{align}
||\Delta_{r}(s)||_{\mathcal{H}_{2,\omega}}&=\sqrt{trace(C_{r}P_{r}C_{r}^T)}\nonumber\\
&=\big(trace(D_iCPC^TD_i^T-2D_iCP_{12}\tilde{C}^TD_i^T+2D_iCP_{13}C_i^T\nonumber\\
&\hspace*{3cm}+D_i\tilde{C}P_{22}\tilde{C}^TD_i^T-2D_i\tilde{C}P_{23}C_i^T+C_iP_{33}C_i^T)\big)^{\frac{1}{2}}\nonumber\\
&=\sqrt{trace(B_{r}^TQ_{r}B_{r})}\nonumber\\
&=\big(trace(B^TQ_{11}B+2B^TQ_{12}\tilde{B}+\tilde{B}^TQ_{22}\tilde{B})\big)^{\frac{1}{2}}.\nonumber
\end{align}
\subsection{Necessary Conditions for the Local optimum of $||\Delta_r(s)||_{\mathcal{H}_{2,\omega}}^2$}
To ensure clarity, let's begin by introducing the variables required for deriving the necessary conditions for a local optimum of $||\Delta_r(s)||_{\mathcal{H}_{2,\omega}}^2$.  We define $\bar{P}_{11}$, $\bar{P}_{12}$, $\bar{P}_{13}$, $\bar{P}_{22}$, $\bar{P}_{33}$, and $E_1$ as the solutions to the following set of linear matrix equations:
\begin{align}
A\bar{P}_{11}+\bar{P}_{11}A^T+BB^T&=0,\hspace*{0.5cm}\label{e23}\\
A\bar{P}_{12}+\bar{P}_{12}\tilde{A}^T+B\tilde{B}^T&=0,\label{e24}\\
A\bar{P}_{13}+\bar{P}_{13}A_i^T+\bar{P}_{11}C^TB_i^T-\bar{P}_{12}\tilde{C}^TB_i^T&=0,\label{e25}\\
\tilde{A}\bar{P}_{22}+\bar{P}_{22}\tilde{A}^T+\tilde{B}\tilde{B}^T&=0,\label{e26}\\
\tilde{A}\bar{P}_{23}+\bar{P}_{23}A_i^T+\bar{P}_{12}^TC^TB_i^T-\bar{P}_{22}\tilde{C}^TB_i^T&=0,\label{e27}\\
A_i\bar{P}_{33}+\bar{P}_{33}A_i^T+B_iC\bar{P}_{13}-B_i\tilde{C}\bar{P}_{23}+\bar{P}_{13}^TC^TB_i^T-\bar{P}_{23}^T\tilde{C}^TB_i^T&=0,\label{e28}\\
A_i^TE_1-E_1A^T+C_i^TD_iC\bar{P}_{11}-C_i^TD_i\tilde{C}\bar{P}_{12}^T+C_i^TC_i\bar{P}_{13}^T&=0.\label{e22}
\end{align}
These variables will be utilized in the subsequent derivation of the partial derivatives of $||\Delta_r(s)||_{\mathcal{H}_{2,\omega}}^2$ with respect to $\tilde{A}$ and $\tilde{B}$.

For the derivation of the partial derivative of $||\Delta_r(s)||_{\mathcal{H}_{2,\omega}}^2$ with respect to $\tilde{C}$, it is necessary to introduce variables $\bar{Q}_{13}$, $\bar{Q}_{23}$, $\bar{Q}_{33}$, and $E_2$, which represent the solutions to the following set of linear matrix equations:
\begin{align}
A^T\bar{Q}_{13}+\bar{Q}_{13}A_i+C^TB_i^T\bar{Q}_{33}+C^TD_i^TC_i&=0,\label{e30}\\
\tilde{A}^T\bar{Q}_{23}+\bar{Q}_{23}A_i-\tilde{C}^TB_i^T\bar{Q}_{33}-\tilde{C}^TD_i^TC_i&=0,\label{e31}\\
A_i^T\bar{Q}_{33}+\bar{Q}_{33}A_i+C_i^TC_i&=0,\label{e32}\\
A_i^TE_2-E_2A^T+\bar{Q}_{13}^TBB^T+\tilde{B}B^T&=0.\label{e33}
\end{align}

Throughout the derivation of the optimality conditions, we will make frequent use of the following trace properties:
\begin{enumerate}
  \item Addition: $trace(H_1+H_2+H_3)=trace(H_1)+trace(H_2)+trace(H_3)$.
  \item Transpose: $trace(H_1^T)=trace(H_1)$.
  \item Cyclic permutation: $trace(H_1H_2H_3)=trace(H_3H_1H_2)=trace(H_2H_3H_1)$.
  \item Partial derivative: $\Delta_{f(H)}^H=trace\Big(\frac{\partial}{\partial H}\big(f(H)\big)(\Delta_{H})^T\Big)$ where $\Delta_{f(H)}^H$ is the first-order derivative of $f(H)$ with respect to $H$, and $\Delta_{H}$ is the differential of $H$.
\end{enumerate}
These properties can be found in \cite{petersen2008matrix}.

Moving forward, we will make frequent use of the following lemma, which is available in \cite{petersson2013nonlinear} and supports our derivation.
\begin{lemma}\label{lemma}
Let $U$ and $V$ be the solutions of the Sylvester equations:
\begin{align}
RU+US+K&=0,\nonumber\\
SV+VR+L&=0.\nonumber
\end{align}
In this case, we have $trace(KV)=trace(LU)$.
\end{lemma}
\begin{proof}
For further details, refer to the proof of Lemma 4.1 in \cite{petersson2013nonlinear}.
\end{proof}

For the sake of convenience in the ensuing derivation, we present the definition of the Fr'{e}chet derivative of the matrix logarithm, which is utilized multiple times.
\begin{definition}
The Fr'{e}chet derivative $L(A,E)$ of the matrix logarithm $ln(A)$ in the direction of matrix $E$ is given by the integral expression:
\begin{align}
L(A,E)=\int_{0}^{1}\big(s(A-I)+I\big)^{-1}E\big(s(A-I)+I\big)^{-1}ds,\nonumber
\end{align}as outlined in \cite{higham2008functions}.
\end{definition}

Finally, we introduce the following variables to conclude the definitions:
\begin{align}
W_1&=-\bar{P}_{12}^TC^TD_i^TC_i+\bar{P}_{13}^TC^TD_i^TC_i+\bar{P}_{22}C_i^TC_i\nonumber\\
&\hspace*{2.65cm}-\bar{P}_{23}^TC_i^TC_i-\bar{P}_{23}C_i^TC_i+\bar{P}_{33}C_i^TC_i,\nonumber\\
W_2&=\tilde{B}B^T\bar{Q}_{13}+\tilde{B}\tilde{B}^T,\nonumber\\
V_1&=Real\Big[\frac{j}{2\pi}L\big(-A_i-j\omega I,W_1\big)\Big],\nonumber\\
V_2&=Real\Big[\frac{j}{2\pi}L\big(-A_i-j\omega I,-B_iCE_1^T\big)\Big],\nonumber\\
V_3&=Real\Big[\frac{j}{2\pi}L\big(-A_i-j\omega I,W_2\big)\Big],\nonumber\\
V_4&=Real\Big[\frac{j}{2\pi}L\big(-A_i-j\omega I,B_iCE_2^T\big)\Big].\nonumber
\end{align}

The theorem presented below outlines the necessary conditions that a local optimum ROM $(\tilde{A},\tilde{B},\tilde{C},D)$ must satisfy.
\begin{theorem}
Assuming that $\tilde{H}(s)$ is a stable minimum phase system and $D$ is full rank, the local optimum $(\tilde{A},\tilde{B},\tilde{C},D)$ of $||\Delta_r(s)||_{\mathcal{H}_{2,\omega}}^2$ must satisfy the following optimality conditions:
\begin{align}
Q_{12}^T\bar{P}_{12}+Q_{22}\bar{P}_{22}+d_1&=0,\label{e34}\\
Q_{12}^TB+Q_{22}\tilde{B}+d_2&=0,\label{e35}\\
-D_i^TD_iCP_{12}+D_i^TD_i\tilde{C}P_{22}+d_3&=0,\label{e36}
\end{align}wherein
\begin{align}
d_1&=Q_{13}^T\bar{P}_{13}+Q_{23}\bar{P}_{23}+Q_{23}\bar{P}_{23}^T+Q_{33}\bar{P}_{33}+E_1S_{\omega,1}^T-V_1^T-V_2^T,\nonumber\\
d_2&=-Q_{13}^T\bar{P}_{11}C^TD_i^T+Q_{13}^T\bar{P}_{12}C_i^T-Q_{13}^T\bar{P}_{13}C_i^T-Q_{23}\bar{P}_{12}^TC^TD_i^T\nonumber\\
&\hspace*{0.5cm}-Q_{23}\bar{P}_{23}C_i^T+Q_{23}\bar{P}_{22}C_i^T-Q_{33}\bar{P}_{13}^TC^TD_i^T+Q_{33}\bar{P}_{23}^TC_i^T\nonumber\\
&\hspace*{0.5cm}-Q_{33}\bar{P}_{33}C_i^T+V_1^TC_i^T-V_2^TC_i^T-E_1S_{\omega,1}^TC_i^T-E_1S_\omega[A^T]C^TD_i^T\nonumber\\
&\hspace*{0.5cm}+S_\omega[A_i^T]E_1C^TD_i^T,\nonumber\\
d_3&=D_i^TD_iCP_{13}-D_i^TD_i\tilde{C}P_{23}-D_i^TC_iP_{23}^T+D_i^TC_iP_{33}+B_i^T\bar{Q}_{13}^TP_{13}\nonumber\\
&\hspace*{0.5cm}-B_i^T\bar{Q}_{13}^TP_{12}+B_i^T\bar{Q}_{23}P_{23}-B_i^T\bar{Q}_{23}P_{22}+B_i^T\bar{Q}_{33}P_{33}-B_i^T\bar{Q}_{33}P_{23}^T\nonumber\\
&\hspace*{0.5cm}+B_i^TE_2S_{\omega,1}^T+B_i^TV_3^T+B_i^TV_4^T.\nonumber
\end{align}
\end{theorem}
\begin{proof}
Let us define the cost function $J$ as $J=||\Delta_r||_{\mathcal{H}_{2,\omega}}^2$ and denote the first-order derivatives of $Q_{11}$, $Q_{12}$, $Q_{13}$, $Q_{22}$, $Q_{23}$, $Q_{33}$, $S_{\omega,1}$, $S_{\omega,2}$, $S_{\omega}[\tilde{A}]$, $S_{\omega}[A_i]$, and $J$ with respect to $\tilde{A}$ as $\Delta_{Q_{11}}^{\tilde{A}}$, $\Delta_{Q_{12}}^{\tilde{A}}$, $\Delta_{Q_{13}}^{\tilde{A}}$, $\Delta_{Q_{22}}^{\tilde{A}}$, $\Delta_{Q_{23}}^{\tilde{A}}$, $\Delta_{Q_{33}}^{\tilde{A}}$, $\Delta_{S_{\omega,1}}^{\tilde{A}}$, $\Delta_{S_{\omega,2}}^{\tilde{A}}$, $\Delta_{S_{\omega,r}}^{\tilde{A}}$, $\Delta_{S_{\omega,i}}^{\tilde{A}}$, and $\Delta_{J}^{\tilde{A}}$, respectively. Further, let us denote the differential of $\tilde{A}$ as $\Delta_{\tilde{A}}$. By differentiating $J$ with respect to $\tilde{A}$, we get
      \begin{align}
      \Delta_{J}^{\tilde{A}}&=trace(B^T\Delta_{Q_{11}}^{\tilde{A}}B+2B^T\Delta_{Q_{12}}^{\tilde{A}}\tilde{B}+\tilde{B}^T\Delta_{Q_{22}}^{\tilde{A}}\tilde{B})\nonumber\\
      &=trace(BB^T\Delta_{Q_{11}}^{\tilde{A}}+2B\tilde{B}^T(\Delta_{Q_{12}}^{\tilde{A}})^T+\tilde{B}\tilde{B}^T\Delta_{Q_{22}}^{\tilde{A}}).\nonumber
      \end{align}
By differentiating equations (\ref{e16}), (\ref{e17}), and (\ref{e19}) with respect to $\tilde{A}$, we get
      \begin{align}
      A^T\Delta_{Q_{11}}^{\tilde{A}}+\Delta_{Q_{11}}^{\tilde{A}}A+R_{11}&=0,\hspace*{2cm}\label{e40}\\
      A^T\Delta_{Q_{12}}^{\tilde{A}}+\Delta_{Q_{12}}^{\tilde{A}}\tilde{A}+R_{12}&=0,\label{e41}\\
      \tilde{A}^T\Delta_{Q_{22}}^{\tilde{A}}+\Delta_{Q_{22}}^{\tilde{A}}\tilde{A}+R_{22}&=0,\label{e42}
      \end{align}
      wherein
      \begin{align}
      \hspace{1cm}R_{11}&=C^TB_i^T(\Delta_{Q_{13}}^{\tilde{A}})^T+\Delta_{Q_{13}}^{\tilde{A}}B_iC+(\Delta_{S_{\omega,1}}^{\tilde{A}})^TC_i^TD_iC+C^TD_i^TC_i\Delta_{S_{\omega,1}}^{\tilde{A}},\nonumber\\
      R_{12}&=Q_{12}\Delta_{\tilde{A}}+C^TB_i^T\Delta_{Q_{23}}^{\tilde{A}}-\Delta_{Q_{13}}^{\tilde{A}}B_i\tilde{C}\nonumber\\
      &\hspace{1cm}-C^TD_i^TD_i\tilde{C}\Delta_{S_{\omega,r}}^{\tilde{A}}-(\Delta_{S_{\omega,1}}^{\tilde{A}})^TC_i^TD_i\tilde{C}+C^TD_i^TC_i\Delta_{S_{\omega,2}}^{\tilde{A}},\nonumber\\
      R_{22}&=(\Delta_{\tilde{A}})^TQ_{22}+Q_{22}\Delta_{\tilde{A}}-\tilde{C}^TB_i^T\Delta_{Q_{23}}^{\tilde{A}}-\Delta_{Q_{23}}^{\tilde{A}}B_i\tilde{C}\nonumber\\
      &\hspace{1cm}+(\Delta_{S_{\omega,r}}^{\tilde{A}})^T\tilde{C}^TD_i^TD_i\tilde{C}+\tilde{C}^TD_i^TD_i\tilde{C}\Delta_{S_{\omega,r}}^{\tilde{A}}\nonumber\\
      &\hspace{1cm}-(\Delta_{S_{\omega,2}}^{\tilde{A}})^TC_i^TD_i\tilde{C}-\tilde{C}^TD_i^TC_i\Delta_{S_{\omega,2}}^{\tilde{A}}.\nonumber
      \end{align}
      By applying Lemma \ref{lemma} on equations (\ref{e40}) and (\ref{e23}), (\ref{e41}) and (\ref{e24}), and (\ref{e42}) and (\ref{e26}), we get
      \begin{align}
      trace(BB^T\Delta_{Q_{11}}^{\tilde{A}})&=trace(R_{11}\bar{P}_{11}),\nonumber\\
      trace\big(B\tilde{B}^T(\Delta_{Q_{12}}^{\tilde{A}})^T\big)&=trace(R_{12}^T\bar{P}_{12}),\nonumber\\
      trace(\tilde{B}\tilde{B}^T\Delta_{Q_{22}}^{\tilde{A}})&=trace(R_{22}\bar{P}_{22}).\nonumber
      \end{align}
      Thus
      \begin{align}
      \Delta_{J}^{\tilde{A}}&=trace\Big(2Q_{12}^T\bar{P}_{12}(\Delta_{\tilde{A}})^T+2Q_{22}\bar{P}_{22}(\Delta_{\tilde{A}})^T+2\bar{P}_{11}C^TB_i^T(\Delta_{Q_{13}}^{\tilde{A}})^T\nonumber\\
      &\hspace{1cm}-2\bar{P}_{12}\tilde{C}^TB_i^T(\Delta_{Q_{13}}^{\tilde{A}})^T+2\bar{P}_{12}^TC^TB_i^T\Delta_{Q_{23}}^{\tilde{A}}-2\bar{P}_{22}\tilde{C}^TB_i^T\Delta_{Q_{23}}^{\tilde{A}}\nonumber\\
      &\hspace{1cm}-2\tilde{C}^TD_i^TD_iC\bar{P}_{12}(\Delta_{S_{\omega,r}}^{\tilde{A}})^T+2\tilde{C}^TD_i^TD_i\tilde{C}\bar{P}_{22}(\Delta_{S_{\omega,r}}^{\tilde{A}})^T\nonumber\\
      &\hspace{1cm}+2C_i^TD_iC\bar{P}_{12}(\Delta_{S_{\omega,2}}^{\tilde{A}})^T-2C_i^TD_i\tilde{C}\bar{P}_{22}(\Delta_{S_{\omega,2}}^{\tilde{A}})^T\nonumber\\
      &\hspace{1cm}+2C_i^TD_iC\bar{P}_{11}(\Delta_{S_{\omega,1}}^{\tilde{A}})^T-2C_i^TD_i\tilde{C}\bar{P}_{12}^T(\Delta_{S_{\omega,1}}^{\tilde{A}})^T\Big).\nonumber
      \end{align}
      By differentiating equations (\ref{e18}), (\ref{e20}), (\ref{e21}), and (\ref{e8}) with respect to $\tilde{A}$, we get
      \begin{align}
      A^T\Delta_{Q_{13}}^{\tilde{A}}+\Delta_{Q_{13}}^{\tilde{A}}A_i+R_{13}&=0,\label{e43}\\
      \tilde{A}^T\Delta_{Q_{23}}^{\tilde{A}}+\Delta_{Q_{23}}^{\tilde{A}}A_i+R_{23}&=0,\label{e44}\\
      A_i^T\Delta_{Q_{33}}^{\tilde{A}}+\Delta_{Q_{33}}^{\tilde{A}}A_i+R_{33}&=0,\label{e45}\\
      A_i\Delta_{S_{\omega,1}}^{\tilde{A}}-\Delta_{S_{\omega,1}}^{\tilde{A}}A+R_1&=0,\label{e46}
      \end{align}
      wherein
      \begin{align}
      R_{13}&=Q_{13}\Delta_{\tilde{A}}+C^TB_i^T\Delta_{Q_{33}}^{\tilde{A}}+C^TD_i^TC_i\Delta_{S_{\omega,i}}^{\tilde{A}}+(\Delta_{S_{\omega,1}}^{\tilde{A}})^TC_i^TC_i,\nonumber\\
      R_{23}&=(\Delta_{\tilde{A}})^TQ_{23}+Q_{23}\Delta_{\tilde{A}}-\tilde{C}^TB_i^T\Delta_{Q_{33}}^{\tilde{A}}-(\Delta_{S_{\omega,r}}^{\tilde{A}})^T\tilde{C}^TD_i^TC_i\nonumber\\
      &\hspace*{2cm}-\tilde{C}^TD_i^TC_i\Delta_{S_{\omega,i}}^{\tilde{A}}+(\Delta_{S_{\omega,2}}^{\tilde{A}})^TC_i^TC_i,\nonumber\\
      R_{33}&=(\Delta_{\tilde{A}})^TQ_{33}+Q_{33}\Delta_{\tilde{A}}+(\Delta_{S_{\omega,i}}^{\tilde{A}})^TC_i^TC_i+C_i^TC_i\Delta_{S_{\omega,i}}^{\tilde{A}},\nonumber\\
      R_1&=\Delta_{\tilde{A}}S_{\omega,1}-\Delta_{S_{\omega,i}}^{\tilde{A}}B_iC.\nonumber
      \end{align}
      By applying Lemma \ref{lemma} on equations (\ref{e43}) and (\ref{e25}), and (\ref{e44}) and (\ref{e27}), we get
      \begin{align}
      trace(R_{13}^T\bar{P}_{13})&=trace\big(\bar{P}_{11}C^TB_i^T(\Delta_{Q_{13}}^{\tilde{A}})^T-\bar{P}_{12}\tilde{C}^TB_i^T(\Delta_{Q_{13}}^{\tilde{A}})^T\big),\nonumber\\
      trace(R_{23}^T\bar{P}_{23})&=trace\big(\bar{P}_{12}^TC^TB_i^T\Delta_{Q_{23}}^{\tilde{A}}-\bar{P}_{22}\tilde{C}^TB_i\Delta_{Q_{23}}^{\tilde{A}}\big).\nonumber
      \end{align}
      Thus
      \begin{align}
      \Delta_{J}^{\tilde{A}}&=trace\Big(2Q_{12}^T\bar{P}_{12}(\Delta_{\tilde{A}})^T+2Q_{22}\bar{P}_{22}(\Delta_{\tilde{A}})^T+2Q_{13}^T\bar{P}_{13}(\Delta_{\tilde{A}})^T\nonumber\\
      &\hspace*{1cm}+2Q_{23}\bar{P}_{23}(\Delta_{\tilde{A}})^T+2Q_{23}\bar{P}_{23}^T(\Delta_{\tilde{A}})^T+2B_iC\bar{P}_{13}\Delta_{Q_{33}}^{\tilde{A}}\nonumber\\
      &\hspace*{1cm}-2B_i\tilde{C}\bar{P}_{23}\Delta_{Q_{33}}^{\tilde{A}}-2\tilde{C}^TD_i^TD_iC\bar{P}_{12}(\Delta_{S_{\omega,r}}^{\tilde{A}})^T\nonumber\\
      &\hspace*{1cm}+2\tilde{C}^TD_i^TD_i\tilde{C}\bar{P}_{22}(\Delta_{S_{\omega,r}}^{\tilde{A}})^T-2\tilde{C}^TD_i^TC_i\bar{P}_{23}^T(\Delta_{S_{\omega,r}}^{\tilde{A}})^T\nonumber\\
      &\hspace*{1cm}+2C_i^TD_iC\bar{P}_{13}(\Delta_{S_{\omega,i}}^{\tilde{A}})^T-2C_i^TD_i\tilde{C}\bar{P}_{23}(\Delta_{S_{\omega,i}}^{\tilde{A}})^T\nonumber\\
      &\hspace*{1cm}+2C_i^TD_iC\bar{P}_{12}(\Delta_{S_{\omega,2}}^{\tilde{A}})^T-2C_i^TD_i\tilde{C}\bar{P}_{22}(\Delta_{S_{\omega,2}}^{\tilde{A}})^T\nonumber\\
      &\hspace*{1cm}+2C_i^TC_i\bar{P}_{23}^T(\Delta_{S_{\omega,2}}^{\tilde{A}})^T+2C_i^TD_iC\bar{P}_{11}(\Delta_{S_{\omega,1}}^{\tilde{A}})^T\nonumber\\
      &\hspace*{1cm}-2C_i^TD_i\tilde{C}\bar{P}_{12}^T(\Delta_{S_{\omega,1}}^{\tilde{A}})^T+2C_i^TC_i\bar{P}_{13}^T(\Delta_{S_{\omega,1}}^{\tilde{A}})^T\Big).\nonumber
      \end{align}
      By applying Lemma \ref{lemma} on equations (\ref{e45}) and (\ref{e28}), we get
      \begin{align}
      trace(R_{33}\bar{P}_{33})=2trace(B_iC\bar{P}_{13}\Delta_{Q_{33}}^{\tilde{A}}-B_i\tilde{C}\bar{P}_{23}\Delta_{Q_{33}}^{\tilde{A}}).\nonumber
      \end{align}
      Thus
      \begin{align}
      \Delta_{J}^{\tilde{A}}&=trace\Big(2Q_{12}^T\bar{P}_{12}(\Delta_{\tilde{A}})^T+2Q_{22}\bar{P}_{22}(\Delta_{\tilde{A}})^T+2Q_{13}^T\bar{P}_{13}(\Delta_{\tilde{A}})^T\nonumber\\
      &\hspace*{0.5cm}+2Q_{23}\bar{P}_{23}(\Delta_{\tilde{A}})^T+2Q_{23}\bar{P}_{23}^T(\Delta_{\tilde{A}})^T+2Q_{33}\bar{P}_{33}(\Delta_{\tilde{A}})^T\nonumber\\
      &\hspace*{0.5cm}-2\tilde{C}^TD_i^TD_iC\bar{P}_{12}(\Delta_{S_{\omega,r}}^{\tilde{A}})^T+2\tilde{C}^TD_i^TD_i\tilde{C}\bar{P}_{22}(\Delta_{S_{\omega,r}}^{\tilde{A}})^T\nonumber\\
      &\hspace*{0.5cm}-2\tilde{C}^TD_i^TC_i\bar{P}_{23}^T(\Delta_{S_{\omega,r}}^{\tilde{A}})^T+2C_i^TD_iC\bar{P}_{13}(\Delta_{S_{\omega,i}}^{\tilde{A}})^T\nonumber\\
      &\hspace*{0.5cm}-2C_i^TD_i\tilde{C}\bar{P}_{23}(\Delta_{S_{\omega,i}}^{\tilde{A}})^T+2C_i^TC_i\bar{P}_{33}(\Delta_{S_{\omega,i}}^{\tilde{A}})^T\nonumber\\
      &\hspace*{0.5cm}+2C_i^TD_iC\bar{P}_{12}(\Delta_{S_{\omega,2}}^{\tilde{A}})^T-2C_i^TD_i\tilde{C}\bar{P}_{22}(\Delta_{S_{\omega,2}}^{\tilde{A}})^T\nonumber\\
      &\hspace*{0.5cm}+2C_i^TC_i\bar{P}_{23}^T(\Delta_{S_{\omega,2}}^{\tilde{A}})^T+2C_i^TD_iC\bar{P}_{11}(\Delta_{S_{\omega,1}}^{\tilde{A}})^T\nonumber\\
      &\hspace*{0.5cm}-2C_i^TD_i\tilde{C}\bar{P}_{12}^T(\Delta_{S_{\omega,1}}^{\tilde{A}})^T+2C_i^TC_i\bar{P}_{13}^T(\Delta_{S_{\omega,1}}^{\tilde{A}})^T\Big).\nonumber
      \end{align}
      Since $S_{\omega,2}=S_\omega[\tilde{A}]-S_\omega[A_i]$, $\Delta_{S_{\omega,2}}^{\tilde{A}}=\Delta_{S_{\omega,r}}^{\tilde{A}}-\Delta_{S_{\omega,i}}^{\tilde{A}}$. Therefore,
      \begin{align}
      \Delta_{J}^{\tilde{A}}&=trace\Big(2Q_{12}^T\bar{P}_{12}(\Delta_{\tilde{A}})^T+2Q_{22}\bar{P}_{22}(\Delta_{\tilde{A}})^T+2Q_{13}^T\bar{P}_{13}(\Delta_{\tilde{A}})^T\nonumber\\
      &\hspace*{0.5cm}+2Q_{23}\bar{P}_{23}(\Delta_{\tilde{A}})^T+2Q_{23}\bar{P}_{23}^T(\Delta_{\tilde{A}})^T+2Q_{33}\bar{P}_{33}(\Delta_{\tilde{A}})^T\nonumber\\
      &\hspace*{0.5cm}+2W_1\Delta_{S_{\omega,i}}^{\tilde{A}}+2C_i^TD_iC\bar{P}_{11}(\Delta_{S_{\omega,1}}^{\tilde{A}})^T-2C_i^TD_i\tilde{C}\bar{P}_{12}^T(\Delta_{S_{\omega,1}}^{\tilde{A}})^T\nonumber\\
      &\hspace*{0.5cm}+2C_i^TC_i\bar{P}_{13}^T(\Delta_{S_{\omega,1}}^{\tilde{A}})^T\Big).\nonumber
      \end{align}
      Note that
      \begin{align}
      \Delta_{S_{\omega,i}}^{\tilde{A}}&=Re\Big[\frac{j}{2\pi}L\big(-A_i-j\omega I,-\Delta_{\tilde{A}}\big)\Big],\nonumber
      \end{align}cf. \cite{petersson2013nonlinear}. Owing to the interchangeability of trace and integral,
      \begin{align}
      trace\big(W_1\Delta_{S_{\omega,i}}^{\tilde{A}}\big)&=-trace(V_1\Delta_{\tilde{A}}).\nonumber
      \end{align}Therefore,
      \begin{align}
      \Delta_{J}^{\tilde{A}}&=trace\Big(2Q_{12}^T\bar{P}_{12}(\Delta_{\tilde{A}})^T+2Q_{22}\bar{P}_{22}(\Delta_{\tilde{A}})^T+2Q_{13}^T\bar{P}_{13}(\Delta_{\tilde{A}})^T\nonumber\\
      &\hspace*{1cm}+2Q_{23}\bar{P}_{23}(\Delta_{\tilde{A}})^T+2Q_{23}\bar{P}_{23}^T(\Delta_{\tilde{A}})^T+2Q_{33}\bar{P}_{33}(\Delta_{\tilde{A}})^T\nonumber\\
      &\hspace*{1cm}-2V_1^T(\Delta_{\tilde{A}})^T+2C_i^TD_iC\bar{P}_{11}(\Delta_{S_{\omega,1}}^{\tilde{A}})^T-2C_i^TD_i\tilde{C}\bar{P}_{12}^T(\Delta_{S_{\omega,1}}^{\tilde{A}})^T\nonumber\\
      &\hspace*{1cm}+2C_i^TC_i\bar{P}_{13}^T(\Delta_{S_{\omega,1}}^{\tilde{A}})^T\Big).\nonumber
      \end{align}
      By applying Lemma \ref{lemma} on equations (\ref{e46}) and (\ref{e22}), we get
      \begin{align}
      trace(R_1^TE_1)&=trace\big(E_1S_{\omega,1}^T(\Delta_{\tilde{A}})^T-E_1C^TB_i^T(\Delta_{S_{\omega,i}}^{\tilde{A}})^T\big).\nonumber
      \end{align}
      Thus
      \begin{align}
      \Delta_{J}^{\tilde{A}}&=trace\Big(2Q_{12}^T\bar{P}_{12}(\Delta_{\tilde{A}})^T+2Q_{22}\bar{P}_{22}(\Delta_{\tilde{A}})^T+2Q_{13}^T\bar{P}_{13}(\Delta_{\tilde{A}})^T\nonumber\\
      &\hspace*{1cm}+2Q_{23}\bar{P}_{23}(\Delta_{\tilde{A}})^T+2Q_{23}\bar{P}_{23}^T(\Delta_{\tilde{A}})^T+2Q_{33}\bar{P}_{33}(\Delta_{\tilde{A}})^T\nonumber\\
      &\hspace*{1cm}-2V_1^T(\Delta_{\tilde{A}})^T+2E_1S_{\omega,1}^T(\Delta_{\tilde{A}})^T-2E_1C^TB_i^T(\Delta_{S_{\omega,i}}^{\tilde{A}})^T\Big).\nonumber
      \end{align}
      Note that $-trace(V_2\Delta_{\tilde{A}})=trace(-B_iCE_1^T\Delta_{S_{\omega,i}}^{\tilde{A}})$ due to the interchangeability of trace and integral. Thus
      \begin{align}
      \Delta_{J}^{\tilde{A}}&=trace\Big(2Q_{12}^T\bar{P}_{12}(\Delta_{\tilde{A}})^T+2Q_{22}\bar{P}_{22}(\Delta_{\tilde{A}})^T+2Q_{13}^T\bar{P}_{13}(\Delta_{\tilde{A}})^T\nonumber\\
      &\hspace*{1cm}+2Q_{23}\bar{P}_{23}(\Delta_{\tilde{A}})^T+2Q_{23}\bar{P}_{23}^T(\Delta_{\tilde{A}})^T+2Q_{33}\bar{P}_{33}(\Delta_{\tilde{A}})^T\nonumber\\
      &\hspace*{1cm}+2E_1S_{\omega,1}^T(\Delta_{\tilde{A}})^T-2V_1^T(\Delta_{\tilde{A}})^T-2V_2^T(\Delta_{\tilde{A}})^T\Big).\nonumber
      \end{align}
      Hence,
      \begin{align}
      \frac{\partial}{\partial\tilde{A}}||\Delta_r(s)||_{\mathcal{H}_{2,\omega}}^2=2(Q_{12}^T\bar{P}_{12}+Q_{22}\bar{P}_{22}+d_1),\nonumber
      \end{align} and
      \begin{align}
      Q_{12}^T\bar{P}_{12}+Q_{22}\bar{P}_{22}+d_1=0\nonumber
      \end{align} is a necessary condition for the local optimum of $||\Delta_r(s)||_{\mathcal{H}_{2,\omega}}^2$.

 Let us denote the first-order derivatives of $Q_{11}$, $Q_{12}$, $Q_{13}$, $Q_{22}$, $Q_{23}$, $Q_{33}$, $S_{\omega,1}$, $S_{\omega,2}$, $S_{\omega}[\tilde{A}]$, $S_{\omega}[A_i]$, and $J$ with respect to $\tilde{B}$ as $\Delta_{Q_{11}}^{\tilde{B}}$, $\Delta_{Q_{12}}^{\tilde{B}}$, $\Delta_{Q_{13}}^{\tilde{B}}$, $\Delta_{Q_{22}}^{\tilde{B}}$, $\Delta_{Q_{23}}^{\tilde{B}}$, $\Delta_{Q_{33}}^{\tilde{B}}$, $\Delta_{S_{\omega,1}}^{\tilde{B}}$, $\Delta_{S_{\omega,2}}^{\tilde{B}}$, $\Delta_{S_{\omega,r}}^{\tilde{B}}$, $\Delta_{S_{\omega,i}}^{\tilde{B}}$, and $\Delta_{J}^{\tilde{B}}$, respectively. Further, let us denote the differential of $\tilde{B}$ as $\Delta_{\tilde{B}}$. By differentiating $J$ with respect to $\tilde{B}$, we get
      \begin{align}
      \Delta_{J}^{\tilde{B}}&=trace(B^T\Delta_{Q_{11}}^{\tilde{B}}B+2B^TQ_{12}\Delta_{\tilde{B}}+2B^T\Delta_{Q_{12}}^{\tilde{B}}\tilde{B}+2\tilde{B}^TQ_{22}\Delta_{\tilde{B}}\nonumber\\
      &\hspace*{3cm}+\tilde{B}^T\Delta_{Q_{22}}^{\tilde{B}}\tilde{B})\nonumber\\
      &=trace(2Q_{12}^TB(\Delta_{\tilde{B}})^T+2Q_{22}\tilde{B}(\Delta_{\tilde{B}})^T+BB^T\Delta_{Q_{11}}^{\tilde{B}}+2B\tilde{B}^T(\Delta_{Q_{12}}^{\tilde{B}})^T\nonumber\\
      &\hspace*{3cm}+\tilde{B}\tilde{B}^T\Delta_{Q_{22}}^{\tilde{B}}).\nonumber
      \end{align}
      By differentiating equations (\ref{e16}), (\ref{e17}), and (\ref{e19}) with respect to $\tilde{B}$, we get
      \begin{align}
      A^T\Delta_{Q_{11}}^{\tilde{B}}+\Delta_{Q_{11}}^{\tilde{B}}A+S_{11}&=0,\label{e47}\\
      A^T\Delta_{Q_{12}}^{\tilde{B}}+\Delta_{Q_{12}}^{\tilde{B}}\tilde{A}+S_{12}&=0,\label{e48}\\
      \tilde{A}^T\Delta_{Q_{22}}^{\tilde{B}}+\Delta_{Q_{22}}^{\tilde{B}}\tilde{A}+S_{22}&=0,\label{e49}
      \end{align}
      wherein
      \begin{align}
      S_{11}&=-C^TD^{-T}(\Delta_{\tilde{B}})^TQ_{13}^T-Q_{13}\Delta_{\tilde{B}}D^{-1}C+C^TB_i^T(\Delta_{Q_{13}}^{\tilde{B}})^T+\Delta_{Q_{13}}^{\tilde{B}}B_iC\nonumber\\
      &\hspace*{1cm}+(\Delta_{S_{\omega,1}}^{\tilde{B}})^TC_i^TD_iC+C^TD_i^TC_i\Delta_{S_{\omega,1}}^{\tilde{B}},\nonumber\\
      S_{12}&=-C^TD^{-T}(\Delta_{\tilde{B}})^TQ_{23}+Q_{13}\Delta_{\tilde{B}}D^{-1}\tilde{C}+C^TB_i^T\Delta_{Q_{23}}^{\tilde{B}}-\Delta_{Q_{13}}^{\tilde{B}}B_i\tilde{C}\nonumber\\
      &\hspace*{1cm}-(\Delta_{S_{\omega,1}}^{\tilde{B}})^TC_i^TD_i\tilde{C}+C^TD_i^TC_i\Delta_{S_{\omega,2}}^{\tilde{B}},\nonumber\\
      S_{22}&=\tilde{C}^TD^{-T}(\Delta_{\tilde{B}})^TQ_{23}+Q_{23}\Delta_{\tilde{B}}D^{-1}\tilde{C}-\tilde{C}^TB_i^T\Delta_{Q_{23}}^{\tilde{B}}-\Delta_{Q_{23}}^{\tilde{B}}B_i\tilde{C}\nonumber\\
      &\hspace*{1cm}-(\Delta_{S_{\omega,2}}^{\tilde{B}})^TC_i^TD_i\tilde{C}-\tilde{C}^TD_i^TC_i\Delta_{S_{\omega,2}}^{\tilde{B}}.\nonumber
      \end{align}
      By applying Lemma \ref{lemma} on equations (\ref{e47}) and (\ref{e23}), (\ref{e48}) and (\ref{e24}), and (\ref{e49}) and (\ref{e26}), we get
      \begin{align}
      trace(BB^T\Delta_{Q_{11}}^{\tilde{B}})&=trace(S_{11}\bar{P}_{11}),\nonumber\\
      trace(B\tilde{B}^T(\Delta_{Q_{12}}^{\tilde{B}})^T)&=trace(S_{12}^T\bar{P}_{12}),\nonumber\\
      trace(\tilde{B}\tilde{B}^T\Delta_{Q_{22}}^{\tilde{B}})&=trace(S_{22}\bar{P}_{22}).\nonumber
      \end{align}
Thus
\begin{align}
      \Delta_{J}^{\tilde{B}}&=trace\Big(2Q_{12}^TB(\Delta_{\tilde{B}})^T+2Q_{22}\tilde{B}(\Delta_{\tilde{B}})^T-2Q_{13}^T\bar{P}_{11}C^TD^{-T}(\Delta_{\tilde{B}})^T\nonumber\\
      &\hspace*{1cm}+2Q_{13}^T\bar{P}_{12}\tilde{C}^TD^{-T}(\Delta_{\tilde{B}})^T-2Q_{23}\bar{P}_{12}^TC^TD^{-T}(\Delta_{\tilde{B}})^T\nonumber\\
      &\hspace*{1cm}+2Q_{23}\bar{P}_{22}\tilde{C}^TD^{-T}(\Delta_{\tilde{B}})^T+2\bar{P}_{11}C^TB_i^T(\Delta_{Q_{13}}^{\tilde{B}})^T\nonumber\\
      &\hspace*{1cm}-2\bar{P}_{12}\tilde{C}^TB_i^T(\Delta_{Q_{13}}^{\tilde{B}})^T+2\bar{P}_{12}^TC^TB_i^T\Delta_{Q_{23}}^{\tilde{B}}-2\bar{P}_{22}\tilde{C}^TB_i^T\Delta_{Q_{23}}^{\tilde{B}}\nonumber\\
      &\hspace*{1cm}+2\bar{P}_{12}^TC^TD_i^TC_i\Delta_{S_{\omega,2}}^{\tilde{B}}-2\bar{P}_{22}\tilde{C}^TD_i^TC_i\Delta_{S_{\omega,2}}^{\tilde{B}}\nonumber\\
      &\hspace*{1cm}+2\bar{P}_{11}C^TD_i^TC_i\Delta_{S_{\omega,1}}^{\tilde{B}}-2\bar{P}_{12}\tilde{C}^TD_i^TC_i\Delta_{S_{\omega,1}}^{\tilde{B}}\Big).\nonumber
      \end{align}
      By differentiating equations (\ref{e18}), (\ref{e20}), and (\ref{e21}) with respect to $\tilde{B}$, we get
      \begin{align}
      A^T\Delta_{Q_{13}}^{\tilde{B}}+\Delta_{Q_{13}}^{\tilde{B}}A_i+S_{13}&=0,\label{e50}\\
      \tilde{A}^T\Delta_{Q_{23}}^{\tilde{B}}+\Delta_{Q_{23}}^{\tilde{B}}A_i+S_{23}&=0,\label{e51}\\
      A_i^T\Delta_{Q_{33}}^{\tilde{B}}+\Delta_{Q_{33}}^{\tilde{B}}A_i+S_{33}&=0,\label{e52}
      \end{align}
      wherein
      \begin{align}
      S_{13}&=-Q_{13}\Delta_{\tilde{B}}D^{-1}\tilde{C}-C^TD^{-T}(\Delta_{\tilde{B}})^TQ_{33}+C^TB_i^T\Delta_{Q_{33}}^{\tilde{B}}+C^TD_i^TC_i\Delta_{S_{\omega,i}}^{\tilde{B}}\nonumber\\
      &\hspace*{1cm}+(\Delta_{S_{\omega,1}}^{\tilde{B}})^TC_i^TC_i,\nonumber\\
      S_{23}&=-Q_{23}\Delta_{\tilde{B}}D^{-1}\tilde{C}+\tilde{C}^TD^{-T}(\Delta_{\tilde{B}})^TQ_{33}-\tilde{C}^TB_i^T\Delta_{Q_{33}}^{\tilde{B}}-\tilde{C}^TD_i^TC_i\Delta_{S_{\omega,i}}^{\tilde{B}}\nonumber\\
      &\hspace*{1cm}+(\Delta_{S_{\omega,2}}^{\tilde{B}})^TC_i^TC_i,\nonumber\\
      S_{33}&=-\tilde{C}^TD^{-T}(\Delta_{\tilde{B}})^TQ_{33}-Q_{33}\Delta_{\tilde{B}}D^{-1}\tilde{C}+(\Delta_{S_{\omega,i}}^{\tilde{B}})^TC_i^TC_i+C_i^TC_i\Delta_{S_{\omega,i}}^{\tilde{B}}.\nonumber
      \end{align}
      By applying Lemma \ref{lemma} on equations (\ref{e50}) and (\ref{e25}), and (\ref{e51}) and (\ref{e27}), we get
      \begin{align}
      trace(\bar{P}_{11}C^TB_i^T(\Delta_{Q_{13}}^{\tilde{B}})^T-\bar{P}_{12}\tilde{C}^TB_i^T(\Delta_{Q_{13}}^{\tilde{B}})^T)&=trace(S_{13}^T\bar{P}_{13}),\nonumber\\
      trace(\bar{P}_{12}^TC^TB_i^T\Delta_{Q_{23}}^{\tilde{B}}-\bar{P}_{22}\tilde{C}^TB_i^T\Delta_{Q_{23}}^{\tilde{B}})&=trace(S_{23}^T\bar{P}_{23}).\nonumber
      \end{align}
      Thus
      \begin{align}
      \Delta_{J}^{\tilde{B}}&=trace\Big(2Q_{12}^TB(\Delta_{\tilde{B}})^T+2Q_{22}\tilde{B}(\Delta_{\tilde{B}})^T-2Q_{13}^T\bar{P}_{11}C^TD^{-T}(\Delta_{\tilde{B}})^T\nonumber\\
      &\hspace*{0.5cm}+2Q_{13}^T\bar{P}_{12}\tilde{C}^TD^{-T}(\Delta_{\tilde{B}})^T-2Q_{13}^T\bar{P}_{13}\tilde{C}^TD^{-T}(\Delta_{\tilde{B}})^T\nonumber\\
      &\hspace*{0.5cm}-2Q_{23}\bar{P}_{12}^TC^TD^{-T}(\Delta_{\tilde{B}})^T-2Q_{23}\bar{P}_{23}\tilde{C}^TD^{-T}(\Delta_{\tilde{B}})^T\nonumber\\
      &\hspace*{0.5cm}+2Q_{23}\bar{P}_{22}\tilde{C}^TD^{-T}(\Delta_{\tilde{B}})^T-2Q_{33}\bar{P}_{13}^TC^TD^{-T}(\Delta_{\tilde{B}})^T\nonumber\\
      &\hspace*{0.5cm}+2Q_{33}\bar{P}_{23}^T\tilde{C}^TD^{-T}(\Delta_{\tilde{B}})^T+2B_iC\bar{P}_{13}\Delta_{Q_{33}}^{\tilde{B}}-2B_i\tilde{C}\bar{P}_{23}\Delta_{Q_{33}}^{\tilde{B}}\nonumber\\
      &\hspace*{0.5cm}+2\bar{P}_{13}^TC^TD_i^TC_i\Delta_{S_{\omega,i}}^{\tilde{B}}-2\bar{P}_{23}^T\tilde{C}^TD_i^TC_i\Delta_{S_{\omega,i}}^{\tilde{B}}+2\bar{P}_{12}^TC^TD_i^TC_i\Delta_{S_{\omega,2}}^{\tilde{B}}\nonumber\\
      &\hspace*{0.5cm}-2\bar{P}_{22}\tilde{C}^TD_i^TC_i\Delta_{S_{\omega,2}}^{\tilde{B}}+2\bar{P}_{23}C_i^TC_i\Delta_{S_{\omega,2}}^{\tilde{B}}+2\bar{P}_{11}C^TD_i^TC_i\Delta_{S_{\omega,1}}^{\tilde{B}}\nonumber\\
      &\hspace*{0.5cm}-2\bar{P}_{12}\tilde{C}^TD_i^TC_i\Delta_{S_{\omega,1}}^{\tilde{B}}+2\bar{P}_{13}C_i^TC_i\Delta_{S_{\omega,1}}^{\tilde{B}}\Big).\nonumber
      \end{align}
      By applying Lemma \ref{lemma} on equations (\ref{e52}) and (\ref{e28}), we get
      \begin{align}
      trace(S_{33}\bar{P}_{33})=trace\big((2B_iC\bar{P}_{13}-2B_i\tilde{C}\bar{P}_{23})\Delta_{Q_{33}}^{\tilde{B}}\big).\nonumber
      \end{align}
      Thus
      \begin{align}
      \Delta_{J}^{\tilde{B}}&=trace\Big(2Q_{12}^TB(\Delta_{\tilde{B}})^T+2Q_{22}\tilde{B}(\Delta_{\tilde{B}})^T-2Q_{13}^T\bar{P}_{11}C^TD^{-T}(\Delta_{\tilde{B}})^T\nonumber\\
      &\hspace*{0.5cm}+2Q_{13}^T\bar{P}_{12}\tilde{C}^TD^{-T}(\Delta_{\tilde{B}})^T-2Q_{13}^T\bar{P}_{13}\tilde{C}^TD^{-T}(\Delta_{\tilde{B}})^T\nonumber\\
      &\hspace*{0.5cm}-2Q_{23}\bar{P}_{12}^TC^TD^{-T}(\Delta_{\tilde{B}})^T-2Q_{23}\bar{P}_{23}\tilde{C}^TD^{-T}(\Delta_{\tilde{B}})^T\nonumber\\
      &\hspace*{0.5cm}+2Q_{23}\bar{P}_{22}\tilde{C}^TD^{-T}(\Delta_{\tilde{B}})^T-2Q_{33}\bar{P}_{13}^TC^TD^{-T}(\Delta_{\tilde{B}})^T\nonumber\\
      &\hspace*{0.5cm}+2Q_{33}\bar{P}_{23}^T\tilde{C}^TD^{-T}(\Delta_{\tilde{B}})^T-2Q_{33}\bar{P}_{33}\tilde{C}^TD^{-T}(\Delta_{\tilde{B}})^T\nonumber\\
      &\hspace*{0.5cm}+2\bar{P}_{13}^TC^TD_i^TC_i\Delta_{S_{\omega,i}}^{\tilde{B}}-2\bar{P}_{23}^T\tilde{C}^TD_i^TC_i\Delta_{S_{\omega,i}}^{\tilde{B}}+2\bar{P}_{33}C_i^TC_i\Delta_{S_{\omega,i}}^{\tilde{B}}\nonumber\\
      &\hspace*{0.5cm}+2\bar{P}_{12}^TC^TD_i^TC_i\Delta_{S_{\omega,2}}^{\tilde{B}}-2\bar{P}_{22}\tilde{C}^TD_i^TC_i\Delta_{S_{\omega,2}}^{\tilde{B}}+2\bar{P}_{23}C_i^TC_i\Delta_{S_{\omega,2}}^{\tilde{B}}\nonumber\\
      &\hspace*{0.5cm}+2\bar{P}_{11}C^TD_i^TC_i\Delta_{S_{\omega,1}}^{\tilde{B}}-2\bar{P}_{12}\tilde{C}^TD_i^TC_i\Delta_{S_{\omega,1}}^{\tilde{B}}+2\bar{P}_{13}C_i^TC_i\Delta_{S_{\omega,1}}^{\tilde{B}}\Big).\nonumber
      \end{align}
      Since $S_{\omega,2}=S_\omega[\tilde{A}]-S_\omega[A_i]$, $\Delta_{S_{\omega,2}}^{\tilde{B}}=-\Delta_{S_{\omega,i}}^{\tilde{B}}$. Therefore,
      \begin{align}
      \Delta_{J}^{\tilde{B}}&=trace\Big(2Q_{12}^TB(\Delta_{\tilde{B}})^T+2Q_{22}\tilde{B}(\Delta_{\tilde{B}})^T-2Q_{13}^T\bar{P}_{11}C^TD^{-T}(\Delta_{\tilde{B}})^T\nonumber\\
      &\hspace*{0.5cm}+2Q_{13}^T\bar{P}_{12}\tilde{C}^TD^{-T}(\Delta_{\tilde{B}})^T-2Q_{13}^T\bar{P}_{13}\tilde{C}^TD^{-T}(\Delta_{\tilde{B}})^T\nonumber\\
      &\hspace*{0.5cm}-2Q_{23}\bar{P}_{12}^TC^TD^{-T}(\Delta_{\tilde{B}})^T-2Q_{23}\bar{P}_{23}\tilde{C}^TD^{-T}(\Delta_{\tilde{B}})^T\nonumber\\
      &\hspace*{0.5cm}+2Q_{23}\bar{P}_{22}\tilde{C}^TD^{-T}(\Delta_{\tilde{B}})^T-2Q_{33}\bar{P}_{13}^TC^TD^{-T}(\Delta_{\tilde{B}})^T\nonumber\\
      &\hspace*{0.5cm}+2Q_{33}\bar{P}_{23}^T\tilde{C}^TD^{-T}(\Delta_{\tilde{B}})^T-2Q_{33}\bar{P}_{33}\tilde{C}^TD^{-T}(\Delta_{\tilde{B}})^T+2W_1\Delta_{S_{\omega,i}}^{\tilde{B}}\nonumber\\
      &\hspace*{0.5cm}+2\bar{P}_{11}C^TD_i^TC_i\Delta_{S_{\omega,1}}^{\tilde{B}}-2\bar{P}_{12}\tilde{C}^TD_i^TC_i\Delta_{S_{\omega,1}}^{\tilde{B}}+2\bar{P}_{13}C_i^TC_i\Delta_{S_{\omega,1}}^{\tilde{B}}\Big).\nonumber
      \end{align}
Note that
      \begin{align}
      \Delta_{S_{\omega,i}}^{\tilde{B}}&=Re\Big[\frac{j}{2\pi}L\big(-A_i-j\omega I,\Delta_{\tilde{B}}D^{-1}\tilde{C}\big)\Big],\nonumber
      \end{align}cf. \cite{petersson2013nonlinear}. Owing to the interchangeability of trace and integral,
      \begin{align}
      trace\big(W_1\Delta_{S_{\omega,i}}^{\tilde{B}}\big)&=trace(V_1\Delta_{\tilde{B}}D^{-1}\tilde{C}).\nonumber
      \end{align}
      Thus
\begin{align}
      \Delta_{J}^{\tilde{B}}&=trace\Big(2Q_{12}^TB(\Delta_{\tilde{B}})^T+2Q_{22}\tilde{B}(\Delta_{\tilde{B}})^T-2Q_{13}^T\bar{P}_{11}C^TD^{-T}(\Delta_{\tilde{B}})^T\nonumber\\
      &\hspace*{0.5cm}+2Q_{13}^T\bar{P}_{12}\tilde{C}^TD^{-T}(\Delta_{\tilde{B}})^T-2Q_{13}^T\bar{P}_{13}\tilde{C}^TD^{-T}(\Delta_{\tilde{B}})^T\nonumber\\
      &\hspace*{0.5cm}-2Q_{23}\bar{P}_{12}^TC^TD^{-T}(\Delta_{\tilde{B}})^T-2Q_{23}\bar{P}_{23}\tilde{C}^TD^{-T}(\Delta_{\tilde{B}})^T\nonumber\\
      &\hspace*{0.5cm}+2Q_{23}\bar{P}_{22}\tilde{C}^TD^{-T}(\Delta_{\tilde{B}})^T-2Q_{33}\bar{P}_{13}^TC^TD^{-T}(\Delta_{\tilde{B}})^T\nonumber\\
      &\hspace*{0.5cm}+2Q_{33}\bar{P}_{23}^T\tilde{C}^TD^{-T}(\Delta_{\tilde{B}})^T-2Q_{33}\bar{P}_{33}\tilde{C}^TD^{-T}(\Delta_{\tilde{B}})^T\nonumber\\
      &\hspace*{0.5cm}+2V_1^T\tilde{C}^TD^{-T}(\Delta_{\tilde{B}})^T+2C_i^TD_iC\bar{P}_{11}(\Delta_{S_{\omega,1}}^{\tilde{B}})^T\nonumber\\
      &\hspace*{0.5cm}-2C_i^TD_i\tilde{C}\bar{P}_{12}^T(\Delta_{S_{\omega,1}}^{\tilde{B}})^T+2C_i^TC_i\bar{P}_{13}^T(\Delta_{S_{\omega,1}}^{\tilde{B}})^T\Big).\nonumber
      \end{align}
By differentiating the equation (\ref{e8}) with respect to $\tilde{B}$, we get
      \begin{align}
      A_i\Delta_{S_{\omega,1}}^{\tilde{B}}-\Delta_{S_{\omega,1}}^{\tilde{B}}A+S_2=0\label{e53}
      \end{align}where
      \begin{align}
      S_2=-\Delta_{\tilde{B}}D^{-1}\tilde{C}S_{\omega,1}-\Delta_{\tilde{B}}D^{-1}CS_{\omega}[A]-\Delta_{S_{\omega,i}}^{\tilde{B}}B_iC+S_\omega[A_i]\Delta_{\tilde{B}}D^{-1}C.\nonumber
      \end{align}
       By applying Lemma \ref{lemma} on equations (\ref{e53}) and (\ref{e22}), we get
      \begin{align}
      trace(S_2^TE_1)&=trace\big(C_i^TD_iC\bar{P}_{11}(\Delta_{S_{\omega,1}}^{\tilde{B}})^T-C_i^TD_i\tilde{C}\bar{P}_{12}^T(\Delta_{S_{\omega,1}}^{\tilde{B}})^T\nonumber\\
      &\hspace*{1cm}+C_i^TC_i\bar{P}_{13}^T(\Delta_{S_{\omega,1}}^{\tilde{B}})^T\big).\nonumber
      \end{align}
      Thus
\begin{align}
      \Delta_{J}^{\tilde{B}}&=trace\Big(2Q_{12}^TB(\Delta_{\tilde{B}})^T+2Q_{22}\tilde{B}(\Delta_{\tilde{B}})^T-2Q_{13}^T\bar{P}_{11}C^TD^{-T}(\Delta_{\tilde{B}})^T\nonumber\\
      &\hspace*{0.5cm}+2Q_{13}^T\bar{P}_{12}\tilde{C}^TD^{-T}(\Delta_{\tilde{B}})^T-2Q_{13}^T\bar{P}_{13}\tilde{C}^TD^{-T}(\Delta_{\tilde{B}})^T\nonumber\\
      &\hspace*{0.5cm}-2Q_{23}\bar{P}_{12}^TC^TD^{-T}(\Delta_{\tilde{B}})^T-2Q_{23}\bar{P}_{23}\tilde{C}^TD^{-T}(\Delta_{\tilde{B}})^T\nonumber\\
      &\hspace*{0.5cm}+2Q_{23}\bar{P}_{22}\tilde{C}^TD^{-T}(\Delta_{\tilde{B}})^T-2Q_{33}\bar{P}_{13}^TC^TD^{-T}(\Delta_{\tilde{B}})^T\nonumber\\
      &\hspace*{0.5cm}+2Q_{33}\bar{P}_{23}^T\tilde{C}^TD^{-T}(\Delta_{\tilde{B}})^T-2Q_{33}\bar{P}_{33}\tilde{C}^TD^{-T}(\Delta_{\tilde{B}})^T\nonumber\\
      &\hspace*{0.5cm}+2V_1^T\tilde{C}^TD^{-T}(\Delta_{\tilde{B}})^T-2E_1S_{\omega,1}^T\tilde{C}^TD^{-T}(\Delta_{\tilde{B}})^T\nonumber\\
      &\hspace*{0.5cm}-2E_1S_\omega[A^T]C^TD^{-T}(\Delta_{\tilde{B}})^T+2S_\omega[A_i^T]E_1C^TD^{-T}(\Delta_{\tilde{B}})^T\nonumber\\
      &\hspace*{0.5cm}+2B_iCE_1^T\Delta_{S_{\omega,i}}^{\tilde{B}}\Big).\nonumber
      \end{align}
Note that $trace(-V_2\Delta_{\tilde{B}}D^{-1}\tilde{C})=trace(B_iCE_1^T\Delta_{S_{\omega,i}}^{\tilde{B}})$ due to the interchangeability of trace and integral. Thus
\begin{align}
      \Delta_{J}^{\tilde{B}}&=trace\Big(2Q_{12}^TB(\Delta_{\tilde{B}})^T+2Q_{22}\tilde{B}(\Delta_{\tilde{B}})^T-2Q_{13}^T\bar{P}_{11}C^TD_i^T(\Delta_{\tilde{B}})^T\nonumber\\
      &\hspace*{0.5cm}+2Q_{13}^T\bar{P}_{12}C_i^T(\Delta_{\tilde{B}})^T-2Q_{13}^T\bar{P}_{13}C_i^T(\Delta_{\tilde{B}})^T\nonumber\\
      &\hspace*{0.5cm}-2Q_{23}\bar{P}_{12}^TC^TD_i^T(\Delta_{\tilde{B}})^T-2Q_{23}\bar{P}_{23}C_i^T(\Delta_{\tilde{B}})^T\nonumber\\
      &\hspace*{0.5cm}+2Q_{23}\bar{P}_{22}C_i^T(\Delta_{\tilde{B}})^T-2Q_{33}\bar{P}_{13}^TC^TD_i^T(\Delta_{\tilde{B}})^T\nonumber\\
      &\hspace*{0.5cm}+2Q_{33}\bar{P}_{23}^TC_i^T(\Delta_{\tilde{B}})^T-2Q_{33}\bar{P}_{33}C_i^T(\Delta_{\tilde{B}})^T\nonumber\\
      &\hspace*{0.5cm}+2V_1^TC_i^T(\Delta_{\tilde{B}})^T-2V_2^TC_i^T(\Delta_{\tilde{B}})^T-2E_1S_{\omega,1}^TC_i^T(\Delta_{\tilde{B}})^T\nonumber\\
      &\hspace*{0.5cm}-2E_1S_\omega[A^T]C^TD_i^T(\Delta_{\tilde{B}})^T+2S_\omega[A_i^T]E_1C^TD_i^T(\Delta_{\tilde{B}})^T\Big).\nonumber
      \end{align}
      Hence,
      \begin{align}
      \frac{\partial}{\partial\tilde{B}}||\Delta_r(s)||_{\mathcal{H}_{2,\omega}}^2=2(Q_{12}^TB+Q_{22}\tilde{B}+d_2),\nonumber
      \end{align} and
      \begin{align}
      Q_{12}^TB+Q_{22}\tilde{B}+d_2=0\nonumber
      \end{align} is a necessary condition for the local optimum of $||\Delta_r(s)||_{\mathcal{H}_{2,\omega}}^2$.

Let us denote the first-order derivatives of $P_{13}$, $P_{23}$, $P_{33}$, $S_{\omega,1}$, $S_{\omega,2}$, $S_{\omega}[\tilde{A}]$, $S_{\omega}[A_i]$ and $J$ with respect to $\tilde{C}$ as $\Delta_{P_{13}}^{\tilde{C}}$, $\Delta_{P_{23}}^{\tilde{C}}$, $\Delta_{P_{33}}^{\tilde{C}}$, $\Delta_{S_{\omega,1}}^{\tilde{C}}$, $\Delta_{S_{\omega,2}}^{\tilde{C}}$, $\Delta_{S_{\omega,r}}^{\tilde{C}}$, $\Delta_{S_{\omega,i}}^{\tilde{C}}$, and $\Delta_{J}^{\tilde{C}}$, respectively. Further, let us denote the differential of $\tilde{C}$ as $\Delta_{\tilde{C}}$. By taking differentiation of $J$ with respect to $\tilde{C}$, we get
          \begin{align}
      \Delta_{J}^{\tilde{C}}&=trace\Big(-2D_iCP_{12}(\Delta_{\tilde{C}})^TD_i^T+2D_iCP_{13}(\Delta_{\tilde{C}})^TD_i^T+2D_i\tilde{C}P_{22}(\Delta_{\tilde{C}})^TD_i^T\nonumber\\
      &\hspace*{2cm}-2D_i\tilde{C}P_{23}(\Delta_{\tilde{C}})^TD_i^T-2C_iP_{23}^T(\Delta_{\tilde{C}})^TD_i^T+2C_iP_{33}(\Delta_{\tilde{C}})^TD_i^T\nonumber\\
      &\hspace*{2cm}+2D_iC\Delta_{P_{13}}^{\tilde{C}}C_i^T-2D_i\tilde{C}\Delta_{P_{23}}^{\tilde{C}}C_i^T+C_i\Delta_{P_{33}}^{\tilde{C}}C_i^T\Big)\nonumber\\
      &=trace\Big(-2D_i^TD_iCP_{12}(\Delta_{\tilde{C}})^T+2D_i^TD_iCP_{13}(\Delta_{\tilde{C}})^T+2D_i^TD_i\tilde{C}P_{22}(\Delta_{\tilde{C}})^T\nonumber\\
      &\hspace*{2cm}-2D_i^TD_i\tilde{C}P_{23}(\Delta_{\tilde{C}})^T-2D_i^TC_iP_{23}^T(\Delta_{\tilde{C}})^T+2D_i^TC_iP_{33}(\Delta_{\tilde{C}})^T\nonumber\\
      &\hspace*{2cm}+2C^TD_i^TC_i(\Delta_{P_{13}}^{\tilde{C}})^T-2\tilde{C}^TD_i^TC_i(\Delta_{P_{23}}^{\tilde{C}})^T+C_i^TC_i\Delta_{P_{33}}^{\tilde{C}}\Big).\nonumber
      \end{align}
By taking differentiation of equations (\ref{e11}), (\ref{e13}), (\ref{e14}), and (\ref{e8}) with respect to $\tilde{C}$, we get
\begin{align}
A\Delta_{P_{13}}^{\tilde{C}}+\Delta_{P_{13}}^{\tilde{C}}A_i^T+T_{13}&=0,\label{e54}\\
\tilde{A}\Delta_{P_{23}}^{\tilde{C}}+\Delta_{P_{23}}^{\tilde{C}}A_i^T+T_{23}&=0,\label{e55}\\
A_i\Delta_{P_{33}}^{\tilde{C}}+\Delta_{P_{33}}^{\tilde{C}}A_i^T+T_{33}&=0,\label{e56}\\
A_i\Delta_{S_{\omega,1}}^{\tilde{C}}-\Delta_{S_{\omega,1}}^{\tilde{C}}A+S_3&=0,\label{e57}
\end{align}wherein
\begin{align}
T_{13}&=-P_{13}(\Delta_{\tilde{C}})^TD^{-T}\tilde{B}^T-P_{12}(\Delta_{\tilde{C}})^TB_i^T+BB^T(\Delta_{S_{\omega,1}}^{\tilde{C}})^T+B\tilde{B}^T(\Delta_{S_{\omega,2}}^{\tilde{C}})^T,\nonumber\\
T_{23}&=-P_{23}(\Delta_{\tilde{C}})^TD^{-T}\tilde{B}^T-P_{22}(\Delta_{\tilde{C}})^TB_i^T+\tilde{B}B^T(\Delta_{S_{\omega,1}}^{\tilde{C}})^T+\tilde{B}\tilde{B}^T(\Delta_{S_{\omega,2}}^{\tilde{C}})^T,\nonumber\\
T_{33}&=-\tilde{B}D^{-1}\Delta_{\tilde{C}}P_{33}-P_{33}(\Delta_{\tilde{C}})^TD^{-T}\tilde{B}^T-B_i\Delta_{\tilde{C}}P_{23}-P_{23}^T(\Delta_{\tilde{C}})^TB_i^T\nonumber\\
&\hspace*{0.5cm}+B_iC\Delta_{P_{13}}^{\tilde{C}}+(\Delta_{P_{13}}^{\tilde{C}})^TC^TB_i^T-B_i\tilde{C}\Delta_{P_{23}}^{\tilde{C}}-(\Delta_{P_{23}}^{\tilde{C}})^T\tilde{C}^TB_i^T,\nonumber\\
S_3&=-\tilde{B}D^{-1}\Delta_{\tilde{C}}S_{\omega,1}-\Delta_{S_{\omega,i}}^{\tilde{C}}B_iC.\nonumber
\end{align}
By applying Lemma \ref{lemma} on equations (\ref{e54}) and (\ref{e30}), (\ref{e55}) and (\ref{e31}), and (\ref{e56}) and (\ref{e32}), we get
\begin{align}
trace(T_{13}^T\bar{Q}_{13})&=trace(C^TB_i^T\bar{Q}_{33}(\Delta_{P_{13}}^{\tilde{C}})^T+C^TD_i^TC_i(\Delta_{P_{13}}^{\tilde{C}})^T),\nonumber\\
trace(T_{23}^T\bar{Q}_{23})&=trace(-\tilde{C}^TB_i^T\bar{Q}_{33}(\Delta_{P_{23}}^{\tilde{C}})^T-\tilde{C}^TD_i^TC_i(\Delta_{P_{23}}^{\tilde{C}})^T),\nonumber\\
trace(T_{33}\bar{Q}_{33})&=trace(C_i^TC_i\Delta_{P_{33}}^{\tilde{C}}).\nonumber
\end{align}
Thus
\begin{align}
\Delta_{J}^{\tilde{C}}&=trace\Big(-2D_i^TD_iCP_{12}(\Delta_{\tilde{C}})^T+2D_i^TD_iCP_{13}(\Delta_{\tilde{C}})^T+2D_i^TD_i\tilde{C}P_{22}(\Delta_{\tilde{C}})^T\nonumber\\
&\hspace*{0.5cm}-2D_i^TD_i\tilde{C}P_{23}(\Delta_{\tilde{C}})^T-2D_i^TC_iP_{23}^T(\Delta_{\tilde{C}})^T+2D_i^TC_iP_{33}(\Delta_{\tilde{C}})^T\nonumber\\
&\hspace*{0.5cm}-2D^{-T}\tilde{B}^T\bar{Q}_{13}^TP_{13}(\Delta_{\tilde{C}})^T-2B_i^T\bar{Q}_{13}^TP_{12}(\Delta_{\tilde{C}})^T-2D^{-T}\tilde{B}^T\bar{Q}_{23}P_{23}(\Delta_{\tilde{C}})^T\nonumber\\
&\hspace*{0.5cm}-2B_i^T\bar{Q}_{23}P_{22}(\Delta_{\tilde{C}})^T-2\tilde{D}^{-T}\tilde{B}^T\bar{Q}_{33}P_{33}(\Delta_{\tilde{C}})^T-2B_i^T\bar{Q}_{33}P_{23}^T(\Delta_{\tilde{C}})^T\nonumber\\
&\hspace*{0.5cm}+2\tilde{B}B^T\bar{Q}_{13}\Delta_{S_{\omega,2}}^{\tilde{C}}+2\tilde{B}\tilde{B}^T\Delta_{S_{\omega,2}}^{\tilde{C}}+2BB^T\bar{Q}_{13}\Delta_{S_{\omega,1}}^{\tilde{C}}+2B\tilde{B}^T\Delta_{S_{\omega,1}}^{\tilde{C}}\Big).\nonumber
      \end{align}
 Since $S_{\omega,2}=S_\omega[\tilde{A}]-S_\omega[A_i]$, $\Delta_{S_{\omega,2}}^{\tilde{C}}=-\Delta_{S_{\omega,i}}^{\tilde{C}}$. Note that
      \begin{align}
      \Delta_{S_{\omega,i}}^{\tilde{C}}&=Re\Big[\frac{j}{2\pi}L\big(-A_i-j\omega I,\tilde{B}D^{-1}\Delta_{\tilde{C}}\big)\Big],\nonumber
      \end{align}cf. \cite{petersson2013nonlinear}. Owing to the interchangeability of trace and integral,
      \begin{align}
      trace\big((\tilde{B}B^T\bar{Q}_{13}+\tilde{B}\tilde{B}^T)\Delta_{S_{\omega,i}}^{\tilde{C}}\big)&=trace(V_3\tilde{B}D^{-1}\Delta_{\tilde{C}}).\nonumber
      \end{align}
Thus
\begin{align}
\Delta_{J}^{\tilde{C}}&=trace\Big(-2D_i^TD_iCP_{12}(\Delta_{\tilde{C}})^T+2D_i^TD_iCP_{13}(\Delta_{\tilde{C}})^T+2D_i^TD_i\tilde{C}P_{22}(\Delta_{\tilde{C}})^T\nonumber\\
&\hspace*{1cm}-2D_i^TD_i\tilde{C}P_{23}(\Delta_{\tilde{C}})^T-2D_i^TC_iP_{23}^T(\Delta_{\tilde{C}})^T+2D_i^TC_iP_{33}(\Delta_{\tilde{C}})^T\nonumber\\
&\hspace*{1cm}-2D^{-T}\tilde{B}^T\bar{Q}_{13}^TP_{13}(\Delta_{\tilde{C}})^T-2B_i^T\bar{Q}_{13}^TP_{12}(\Delta_{\tilde{C}})^T\nonumber\\
&\hspace*{1cm}-2D^{-T}\tilde{B}^T\bar{Q}_{23}P_{23}(\Delta_{\tilde{C}})^T-2B_i^T\bar{Q}_{23}P_{22}(\Delta_{\tilde{C}})^T\nonumber\\
&\hspace*{1cm}-2\tilde{D}^{-T}\tilde{B}^T\bar{Q}_{33}P_{33}(\Delta_{\tilde{C}})^T-2B_i^T\bar{Q}_{33}P_{23}^T(\Delta_{\tilde{C}})^T\nonumber\\
&\hspace*{1cm}-2D^{-T}\tilde{B}^TV_3^T(\Delta_{\tilde{C}})^T+2\bar{Q}_{13}^TBB^T(\Delta_{S_{\omega,1}}^{\tilde{C}})^T+2\tilde{B}B^T(\Delta_{S_{\omega,1}}^{\tilde{C}})^T\Big).\nonumber
      \end{align}
By applying Lemma \ref{lemma} on equations (\ref{e57}) and (\ref{e33}), we get
\begin{align}
trace(S_3^TE_2)&=trace\big((\bar{Q}_{13}^TBB^T+\tilde{B}B^T)(\Delta_{S_{\omega,1}}^{\tilde{C}})^T\big).\nonumber
\end{align}
Thus
\begin{align}
\Delta_{J}^{\tilde{C}}&=trace\Big(-2D_i^TD_iCP_{12}(\Delta_{\tilde{C}})^T+2D_i^TD_iCP_{13}(\Delta_{\tilde{C}})^T+2D_i^TD_i\tilde{C}P_{22}(\Delta_{\tilde{C}})^T\nonumber\\
&\hspace*{1cm}-2D_i^TD_i\tilde{C}P_{23}(\Delta_{\tilde{C}})^T-2D_i^TC_iP_{23}^T(\Delta_{\tilde{C}})^T+2D_i^TC_iP_{33}(\Delta_{\tilde{C}})^T\nonumber\\
&\hspace*{1cm}-2D^{-T}\tilde{B}^T\bar{Q}_{13}^TP_{13}(\Delta_{\tilde{C}})^T-2B_i^T\bar{Q}_{13}^TP_{12}(\Delta_{\tilde{C}})^T\nonumber\\
&\hspace*{1cm}-2D^{-T}\tilde{B}^T\bar{Q}_{23}P_{23}(\Delta_{\tilde{C}})^T-2B_i^T\bar{Q}_{23}P_{22}(\Delta_{\tilde{C}})^T\nonumber\\
&\hspace*{1cm}-2\tilde{D}^{-T}\tilde{B}^T\bar{Q}_{33}P_{33}(\Delta_{\tilde{C}})^T-2B_i^T\bar{Q}_{33}P_{23}^T(\Delta_{\tilde{C}})^T\nonumber\\
&\hspace*{1cm}-2D^{-T}\tilde{B}^TV_3^T(\Delta_{\tilde{C}})^T-2D^{-T}\tilde{B}^TE_2S_{\omega,1}^T(\Delta_{\tilde{C}})^T\nonumber\\
&\hspace*{1cm}-2B_iCE_2^T\Delta_{S_{\omega,i}}^{\tilde{C}}\Big).\nonumber
      \end{align}
Note that $trace(V_4\tilde{B}D^{-1}\Delta_{\tilde{C}})=trace(B_iCE_2^T\Delta_{S_{\omega,i}}^{\tilde{C}})$ due to the interchangeability of trace and integral. Thus
\begin{align}
\Delta_{J}^{\tilde{C}}&=trace\Big(-2D_i^TD_iCP_{12}(\Delta_{\tilde{C}})^T+2D_i^TD_i\tilde{C}P_{22}(\Delta_{\tilde{C}})^T+2D_i^TD_iCP_{13}(\Delta_{\tilde{C}})^T\nonumber\\
&\hspace*{1cm}-2D_i^TD_i\tilde{C}P_{23}(\Delta_{\tilde{C}})^T-2D_i^TC_iP_{23}^T(\Delta_{\tilde{C}})^T+2D_i^TC_iP_{33}(\Delta_{\tilde{C}})^T\nonumber\\
&\hspace*{1cm}+2B_i^T\bar{Q}_{13}^TP_{13}(\Delta_{\tilde{C}})^T-2B_i^T\bar{Q}_{13}^TP_{12}(\Delta_{\tilde{C}})^T+2B_i^T\bar{Q}_{23}P_{23}(\Delta_{\tilde{C}})^T\nonumber\\
&\hspace*{1cm}-2B_i^T\bar{Q}_{23}P_{22}(\Delta_{\tilde{C}})^T+2B_i^T\bar{Q}_{33}P_{33}(\Delta_{\tilde{C}})^T-2B_i^T\bar{Q}_{33}P_{23}^T(\Delta_{\tilde{C}})^T\nonumber\\
&\hspace*{1cm}+2B_i^TV_3^T(\Delta_{\tilde{C}})^T+2B_i^TE_2S_{\omega,1}^T(\Delta_{\tilde{C}})^T+2B_i^TV_4^T(\Delta_{\tilde{C}})^T\Big).\nonumber
      \end{align}
Hence,
      \begin{align}
      \frac{\partial}{\partial\tilde{C}}||\Delta_r(s)||_{\mathcal{H}_{2,\omega}}^2=2(-D_i^TD_iCP_{12}+D_i^TD_i\tilde{C}P_{22}+d_3),\nonumber
      \end{align} and
      \begin{align}
      -D_i^TD_iCP_{12}+D_i^TD_i\tilde{C}P_{22}+d_3=0\nonumber
      \end{align} is a necessary condition for the local optimum of $||\Delta_r(s)||_{\mathcal{H}_{2,\omega}}^2$. This completes the proof.
\end{proof}
\subsection{An Oblique Projection Algorithm}
In this subsection, a sketch of an oblique projection algorithm is presented, utilizing the necessary conditions for a local optimum. Furthermore, it addresses the computational and numerical aspects associated with the proposed algorithm.
\subsubsection{Limitation in the Oblique Projection Framework}
Assuming the invertibility of $P_{22}$ and $Q_{22}$, the necessary conditions yield the optimal values for $\tilde{B}$ and $\tilde{C}$ as follows:
\begin{align}
\tilde{B}&=-Q_{22}^{-1}Q_{12}^TB-Q_{22}^{-1}d_2,&\tilde{C}&=CP_{12}P_{22}^{-1}-(D_i^TD_i)^{-1}d_3P_{22}^{-1}.\nonumber
\end{align} Within the oblique projection framework, this implies setting the reduction matrices as $\tilde{V}=P_{12}P_{22}^{-1}$ and $\tilde{W}=-Q_{12}Q_{22}^{-1}$, resulting in:
\begin{align}
\tilde{B}&=\tilde{W}^TB=-Q_{22}^{-1}Q_{12}^TB,&\tilde{C}&=C\tilde{V}=CP_{12}P_{22}^{-1}.\nonumber
\end{align}
Consequently, it follows that:
\begin{align}
Q_{12}^TP_{12}+Q_{22}P_{22}=0\nonumber
\end{align} since $\tilde{W}^T\tilde{V}=I$ in the oblique projection. Before delving into the analysis of deviations from the optimality conditions, let us present a slightly modified expression for the optimality condition (\ref{e34}). Note that:
\begin{align}
Q_{12}^TP_{12}&=Q_{12}^T\big(S_\omega[A]\bar{P}_{12}+\bar{P}_{12}S_\omega[\tilde{A}^T]\big),\nonumber\\
Q_{22}P_{22}&=Q_{22}\big(S_\omega[\tilde{A}]\bar{P}_{22}+\bar{P}_{22}S_\omega[\tilde{A}^T]\big),\nonumber
\end{align}as described in \cite{zulfiqar2019adaptive,zulfiqar2021h2}. By post-multiplying with $S_\omega[\tilde{A}^T]$ (assuming $S_\omega[\tilde{A}^T]$ is full rank), equation (\ref{e34}) becomes:
\begin{align}
&Q_{12}^T\bar{P}_{12}S_\omega[\tilde{A}^T]+Q_{22}\bar{P}_{22}S_\omega[\tilde{A}^T]+d_1S_\omega[\tilde{A}^T]\nonumber\\
=&Q_{12}^TP_{12}+Q_{22}P_{22}-Q_{12}^TS_\omega[A]\bar{P}_{12}-Q_{22}S_\omega[\tilde{A}]\bar{P}_{22}+d_1S_\omega[\tilde{A}^T]\nonumber\\
=&Q_{12}^TP_{12}+Q_{22}P_{22}+\tilde{d}_1\nonumber\\
=&0.\nonumber
\end{align}
It is important to note that $\tilde{d}_1$, $d_2$, and $d_3$ do not reduce to zeros when $\tilde{V}$ and $\tilde{W}$ are set to $\tilde{V}=P_{12}P^{-1}_{22}$ and $\tilde{W}=-Q_{12}Q^{-1}_{22}$.  Hence, it is inherently impossible to construct a local optimum of $||\Delta_r(s)||_{\mathcal{H}_{2,\omega}}^2$ within the oblique projection framework. Nonetheless, the reduction matrices targeting the local optimum are identified. Subsequently, a sketch of an oblique projection algorithm based on these reduction matrices is presented.
\subsubsection{Algorithm}
The matrices associated with the identified reduction matrices, namely $P_{12}$, $P_{22}$, $Q_{12}$, and $Q_{22}$, depend on the unknown ROM. Therefore, an iterative solution starting with an arbitrarily chosen ROM appears to be a viable option. According to the oblique projection theory, it is known that $\tilde{V}=P_{12}$ and $\tilde{W}=Q_{12}$ span the same subspace as $\tilde{V}=P_{12}P^{-1}_{22}$ and $\tilde{W}=-Q_{12}Q^{-1}_{22}$ since the matrices $P^{-1}_{22}$ and $-Q^{-1}_{22}$ only modify the subspace basis \cite{gallivan2004sylvester}. The invertibility of $P_{22}$ and $Q_{22}$ in each iteration can pose a restrictive condition leading to numerical ill-conditioning or even failure of the iterative algorithm. Hence, we proceed with the choice $\tilde{V}=P_{12}$ and $\tilde{W}=Q_{12}$ for the remainder of our discussion. The ROM $(\tilde{A},\tilde{B},\tilde{C})$ and the reduction matrices $(\tilde{V},\tilde{W})$ within the oblique projection framework can be viewed as two interconnected systems given by:
\begin{align}
f(\tilde{V},\tilde{W})&=(\tilde{A},\tilde{B},\tilde{C})&\textnormal{and}&&g(\tilde{A},\tilde{B},\tilde{C})&=(\tilde{V},\tilde{W}).\nonumber
\end{align}Thus, the problem can be considered as a stationary point problem:
\begin{align}
(\tilde{A},\tilde{B},\tilde{C})=f\big(g(\tilde{A},\tilde{B},\tilde{C})\big)\nonumber
\end{align} subject to the constraint $\tilde{W}^T\tilde{V}=I$. Prior to outlining a stationary point iteration algorithm, two issues need to be addressed when $D$ and $\tilde{H}(s)$ are non-invertible. It is worth noting that these issues have been successfully handled in BST, and we can employ the same conventional wisdom to overcome them in our algorithm. In the case of a rank-deficient matrix $D$, it can be replaced with $D=\epsilon I$, where $\epsilon$ is a small positive value \cite{benner2001efficient}. As long as $\epsilon$ is sufficiently small, it does not adversely affect the accuracy when the original $D$ matrix is reintroduced into the final ROM. This strategy is utilized in MATLAB's built-in function ``bst()" for BST to handle rank-deficient $D$ matrices. The stationary point iteration algorithm can become significantly restrictive if it needs to ensure the stability and minimum phase property of $\tilde{H}(s)$ in each iteration. In BST, it has been established that to reduce $||H^{-1}(s)\Delta_a(s)||_{\mathcal{H}_\infty}$, $H(s)$ can be replaced with a minimum phase spectral factor of $H(s)H^*(s)$. We will now demonstrate that the expression of $||\Delta_r(s)||_{\mathcal{H}_{2,\omega}}$ allows us to replace $\tilde{H}(s)$ with a minimum phase spectral factor of $\tilde{H}(s)\tilde{H}^*(s)$.

Consider $\tilde{G}(s)$ as a minimum phase right spectral factor of $\tilde{H}(s)\tilde{H}^*(s)$, guaranteeing $\tilde{G}^*(s)\tilde{G}(s)=\tilde{H}(s)\tilde{H}^*(s)$. By utilizing the definition of the $\mathcal{H}_{2,\omega}$ norm, we can express $||\Delta_r(s)||{\mathcal{H}_{2,\omega}}$ as follows:
\begin{align}
||\Delta_r(s)||_{\mathcal{H}_{2,\omega}}&=\sqrt{\frac{1}{2\pi}\int_{-\Omega}^{\Omega}trace\Big(\Delta_r^*(j\nu)\Delta_r(j\nu)\Big)d\nu}\nonumber\\
  &=\sqrt{\frac{1}{2\pi}\int_{-\Omega}^{\Omega}trace\Big(\Delta_a^*(j\nu)[\tilde{H}(j\nu)\tilde{H}^{-*}(j\nu)]^{-1}\Delta_a(j\nu)\Big)d\nu}\nonumber\\
   &=\sqrt{\frac{1}{2\pi}\int_{-\Omega}^{\Omega}trace\Big(\Delta_a^*(j\nu)\tilde{G}^{-1}(j\nu)\tilde{G}^{-*}(j\nu)\Delta_a(j\nu)\Big)d\nu}\nonumber.
\end{align}
It is worth noting that replacing $\tilde{H}^{-1}(s)$ with $\tilde{G}^{-*}(s)$ has no effect on $||\Delta_r(s)||_{\mathcal{H}_{2,\omega}}$. This substitution offers the advantage of constructing a stable minimum phase $\tilde{G}^*(s)$, even if $\tilde{H}(s)$ is not a minimum phase system. We will now proceed to construct a state-space realization of $\tilde{G}^{-*}(s)$ using a methodology similar to BST.

We can express the controllability gramian $\tilde{P}$ associated with the pair $(\tilde{A},\tilde{B})$ and compute it by solving the Lyapunov equation
   \begin{align}
   \tilde{A}\tilde{P}+\tilde{P}\tilde{A}^T+\tilde{B}\tilde{B}^T=0.\nonumber
   \end{align}
Now, let's introduce the definitions of $\tilde{B}_s$ and $\tilde{A}_s$ as follows:
  \begin{align}
  \tilde{B}_s&=\tilde{P}\tilde{C}^T-\tilde{B}D^T,&\tilde{A}_s&=\tilde{A}-\tilde{B}_s(DD^T)^{-1}\tilde{C}.\nonumber
  \end{align}
Additionally, we need to solve the Riccati equation with $\tilde{X}_s$ given by:
  \begin{align}
  \tilde{A}_s^T\tilde{X}_s+\tilde{X}_s\tilde{A}_s+\tilde{X}_s\tilde{B}_s(DD^T)^{-1}\tilde{B}_s^T\tilde{X}_s+\tilde{C}^T(DD^T)^{-1}\tilde{C}=0.\label{e37}
  \end{align}
 Using these definitions, we can derive a minimum phase realization of $\tilde{G}(s)$ as follows:
  \begin{align}
  \tilde{A}_x&=\tilde{A},&\tilde{B}_x&=\tilde{B}_s,&\tilde{C}_x&=D^{-1}(\tilde{C}-\tilde{B}_s^T\tilde{X}_s),&\tilde{D}_x=D^T\label{e38}
  \end{align}(refer to \cite{zhou1996robust}).
 Moreover, we can obtain a possible realization of $\tilde{G}^*(s)$ as follows:
  \begin{align}
  \tilde{A}_g&=-\tilde{A}_x^T,&\tilde{B}_g&=-\tilde{C}_x^T,&\tilde{C}_g&=\tilde{B}_x^T,&\tilde{D}_g&=\tilde{D}_x^T.
  \end{align}
  Lastly, a stable realization of $\tilde{G}^{-*}(s)$ can be obtained through the expressions:
  \begin{align}
\tilde{A}_{\xi}&=\tilde{A}_g-\tilde{B}_g\tilde{D}_g^{-1}\tilde{C}_g,& \tilde{B}_{\xi}&=-\tilde{B}_g\tilde{D}_g^{-1}, & \tilde{C}_{\xi}&=\tilde{D}_g^{-1}\tilde{C}_g,& \tilde{D}_{\xi}&=\tilde{D}_g^{-1}\label{e39}
\end{align}(refer to \cite{zhou1995frequency,zhou1996robust}).

We present our proposed algorithm, FLRHMORA (Frequency-limited Relative-error $\mathcal{H}_{2}$ MOR Algorithm). The pseudo-code for FLRHMORA is provided in Algorithm \ref{Alg1}. To begin, Step \ref{stp1} replaces a rank-deficient $D$ with a full-rank $D$. Following that, Steps \ref{stp3}-\ref{stp4} involve the computation of a stable state-space realization of $\tilde{G}^{-*}(s)$. Steps \ref{stp5}-\ref{stp7} focus on the calculation of the reduction matrices $\tilde{V}$ and $\tilde{W}$. The bi-orthogonal Gram-Schmidt method, described in \cite{benner2011sparse}, is employed in Steps \ref{stp8}-\ref{stp12} to ensure $\tilde{W}^T\tilde{V}=I$. Finally, Step \ref{stp13} executes an update to the ROM.
\begin{algorithm}[!h]
\textbf{Input:} Original system: $(A,B,C,D)$; Desired frequency interval: $[0,\Omega]$ rad/sec; Initial guess: $(\tilde{A},\tilde{B},\tilde{C})$; Allowable iteration: $i_{max}$.\\ \textbf{Output:} ROM $(\tilde{A},\tilde{B},\tilde{C})$.
  \begin{algorithmic}[1]
      \STATE \textbf{if} ($rank[D]<m$) $D=\epsilon I$ \textbf{end if}\label{stp1}
      \STATE $k=0$, \textbf{while} (not converged) \textbf{do} $k=k+1$.
      \STATE Solve equation (\ref{e37}) to compute $\tilde{X}_s$.\label{stp3}
      \STATE Compute $(\tilde{A}_{\xi},\tilde{B}_{\xi},\tilde{C}_{\xi},\tilde{D}_{\xi})$ using equations (\ref{e38})-(\ref{e39}).\label{stp4}
       \STATE Solve equation (\ref{e6}) to compute $P_{12}$.\label{stp5}
       \STATE Solve the following equations to compute $S_{\omega,1}$ and $S_{\omega,2}$:\\
$\tilde{A}_{\xi}S_{\omega,1}-S_{\omega,1}A+\tilde{B}_{\xi}CS_\omega[A]-S_\omega[\tilde{A}_{\xi}]\tilde{B}_{\xi}C=0$,\\
$\tilde{A}_{\xi}S_{\omega,2}-S_{\omega,2}\tilde{A}+\tilde{B}_{\xi}\tilde{C}S_\omega[\tilde{A}]-S_\omega[\tilde{A}_{\xi}]\tilde{B}_{\xi}\tilde{C}=0$.\\
      \STATE Solve the following equations to compute $Q_{33}$, $Q_{13}$, $Q_{23}$, and $Q_{12}$:\label{stp7}\\
      $\tilde{A}_{\xi}^TQ_{33}+Q_{33}\tilde{A}_{\xi}+S_\omega[\tilde{A}_{\xi}^T]\tilde{C}_{\xi}^T\tilde{C}_{\xi}+\tilde{C}_{\xi}^T\tilde{C}_{\xi}S_\omega[\tilde{A}_{\xi}]=0$,\\
      $A^TQ_{13}+Q_{13}\tilde{A}_{\xi}+C^T\tilde{B}_{\xi}^TQ_{33}+S_\omega[A^T]C^T\tilde{D}_{\xi}^T\tilde{C}_{\xi}+C^T\tilde{D}_{\xi}^T\tilde{C}_{\xi}S_\omega[\tilde{A}_{\xi}]+S_{\omega,1}^T\tilde{C}_{\xi}^T\tilde{C}_{\xi}=0$,\\
      $\tilde{A}^TQ_{23}+Q_{23}\tilde{A}_{\xi}-\tilde{C}^T\tilde{B}_{\xi}^TQ_{33}-S_\omega[\tilde{A}^T]\tilde{C}^T\tilde{D}_{\xi}^T\tilde{C}_{\xi}-\tilde{C}^T\tilde{D}_{\xi}^T\tilde{C}_{\xi}S_\omega[\tilde{A}_{\xi}]+S_{\omega,2}^T\tilde{C}_{\xi}^T\tilde{C}_{\xi}=0$,\\
      $A^TQ_{12}+Q_{12}\tilde{A}+C^T\tilde{B}_{\xi}^TQ_{23}^T-Q_{13}\tilde{B}_{\xi}\tilde{C}-S_\omega[A^T]C^T\tilde{D}_{\xi}^T\tilde{D}_{\xi}\tilde{C}-C^T\tilde{D}_{\xi}^T\tilde{D}_{\xi}\tilde{C}S_\omega[\tilde{A}]-S_{\omega,1}^TC_{\xi}^TD_{\xi}\tilde{C}+C^TD_{\xi}^TC_{\xi}S_{\omega,2}=0$.
      \STATE \textbf{for} $i=1,\ldots,r$ \textbf{do}\label{stp8}
      \STATE $\tilde{v}=P_{12}(:,i)$, $\tilde{v}=\prod_{l=1}^{i}\big(I-P_{12}(:,l)Q_{12}(:,l)^T\big)\tilde{v}$.
      \STATE $\tilde{w}=Q_{12}(:,i)$, $\tilde{w}=\prod_{l=1}^{i}\big(I-Q_{12}(:,l)P_{12}(:,l)^T\big)\tilde{w}$.
      \STATE $\tilde{v}=\frac{\tilde{v}}{||\tilde{v}||_2}$, $\tilde{w}=\frac{\tilde{w}}{||\tilde{w}||_2}$, $\tilde{v}=\frac{\tilde{v}}{\tilde{w}^T\tilde{v}}$, $\tilde{V}(:,i)=\tilde{v}$, $\tilde{W}(:,i)=\tilde{w}$.
      \STATE \textbf{end for}\label{stp12}
      \STATE $\tilde{A}=\tilde{W}^TA\tilde{V}$, $\tilde{B}=\tilde{W}^TB$, $\tilde{C}=C\tilde{V}$.\label{stp13}
      \STATE \textbf{if} ($k=i_{max}$) Break loop \textbf{end if}
      \STATE \textbf{end while}
  \end{algorithmic}
  \caption{FLRHMORA}\label{Alg1}
\end{algorithm}
\begin{remark}
By referring to \cite{petersson2013nonlinear}, it can be noted that FLRHMORA allows the inclusion of a specific frequency range $[\omega_1,\omega_2]$ rad/sec through the adjustment of $S_\omega[A]$, $S_\omega[\tilde{A}]$, and $S_\omega[\tilde{A}_{\xi}]$, given by:
\begin{align}
S_\omega[A]&=Real\Big(\frac{j}{2\pi}ln\big((j\omega_1 I+A)^{-1}(j\omega_2 I+A)\big)\Big),\nonumber\\
S_\omega[\tilde{A}]&=Real\Big(\frac{j}{2\pi}ln\big((j\omega_1 I+\tilde{A})^{-1}(j\omega_2 I+\tilde{A})\big)\Big),\nonumber\\
S_\omega[\tilde{A}_{\xi}]&=Real\Big(\frac{j}{2\pi}ln\big((j\omega_1 I+\tilde{A}_{\xi})^{-1}(j\omega_2 I+\tilde{A}_{\xi})\big)\Big).\nonumber
\end{align}
\end{remark}
\begin{remark}
Considering that FLRHMORA typically does not yield a local optimum even at the point of convergence, it is a reasonable choice to stop the algorithm prematurely after a predefined number of iterations. By doing so, the computational cost can be controlled in cases where FLRHMORA does not converge rapidly.
\end{remark}
\subsubsection{Selection of the Initial Guess}
The initial guess plays a crucial role in the overall performance of an iterative method. As $\Delta_r$ can be interpreted as a weighted measure of additive error, it is recommended to adhere to the guidelines outlined in \cite{anic2013interpolatory,breiten2015near,zulfiqar2022frequency} for selecting the initial guess. These guidelines propose that an initial ROM with dominant poles and/or zeros from the original model leads to a smaller initial error. The computation of such an initial ROM can be efficiently performed using the eigensolvers introduced in \cite{rommes2006efficient} and \cite{martins2007computation}.
\subsubsection{Computational Aspects}
Let us briefly discuss the computational aspects of FLRHMORA.
The computation of matrices $\tilde{X}_s$ and $(\tilde{A}_{\xi},\tilde{B}_{\xi},\tilde{C}_{\xi},\tilde{D}_{\xi})$ is inexpensive due to their involvement with small matrices. Similarly, the matrix logarithms $S_\omega[\tilde{A}]$, $S_\omega[\tilde{A}_{\xi}]$, and $S_{\omega,2}$ can also be computed at a low cost for the same reason. Furthermore, the computation of $Q_{33}$ and $Q_{23}$ is also efficient. However, when dealing with a large-scale matrix $A$, the computation of the matrix logarithm $S_\omega[A]$ becomes computationally expensive. It should be noted that in order to compute $P_{12}$, $S_{\omega,1}$, $Q_{13}$, and $Q_{12}$, we require the products $S_\omega[A]B$ and $CS_\omega[A]$. In \cite{benner2016frequency}, efficient projection methods are proposed to approximate $S_\omega[A]B$ and $CS_\omega[A]$ as $\tilde{V}S_\omega[\tilde{A}]\tilde{B}$ and $\tilde{C}S_\omega[\tilde{A}]\tilde{W}^T$, respectively, specifically designed for large-scale settings. These methods enable the computational feasibility of FLRHMORA in large-scale scenarios.

The most computationally demanding step following the computation of logarithmic products $S_\omega[A]B$ and $CS_\omega[A]$ is the solution of Sylvester equations for $P_{12}$, $S_{\omega,1}$, $Q_{13}$, and $Q_{12}$. These equations have the following structure:
\begin{align}
\mathscr{K}\mathscr{J}+\mathscr{J}\mathscr{L}+\mathscr{M}\mathscr{N}=0,\nonumber
\end{align} where $\mathscr{K}\in\mathbb{R}^{n\times n}$, $\mathscr{L}\in\mathbb{R}^{r\times r}$, $\mathscr{M}\in\mathbb{R}^{n\times q}$, $\mathscr{N}\in\mathbb{R}^{q\times r}$, and $q\ll n$. The large matrices $\mathscr{K}$ and $\mathscr{M}$ are usually sparse, while the smaller matrices $\mathscr{L}$ and $\mathscr{N}$ are dense. The sparsity arises from the mathematical modeling stage of large-scale systems. This specific type of Sylvester equation is known as the ``sparse-dense" Sylvester equation in the literature, and it commonly arises in $\mathcal{H}_2$ MOR algorithms. In \cite{benner2011sparse}, an efficient and robust method is proposed to solve the sparse-dense Sylvester equation, which can further reduce the computational time by utilizing sparse linear matrix solvers, such as those described in \cite{davis2004algorithm} and \cite{demmel1999supernodal}. Furthermore, the computational cost of solving a sparse-dense Sylvester equation is nearly equivalent to that of constructing a rational Krylov subspace, a well-established approach known for its efficiency, as discussed in \cite{wolf2014h,panzer2014model}.

To summarize, FLRHMORA can be effectively applied in large-scale contexts by utilizing the approximation techniques presented in \cite{benner2016frequency} for the logarithmic products $S_\omega[A]B$ and $CS_\omega[A]$, along with the method proposed in \cite{benner2011sparse} for solving sparse-dense Sylvester equations.
\section{Numerical Results}
In this section, the efficacy of FLRHMORA is tested by designing reduced-order controllers for three high-order plants. The first two plants are SISO models while the third plant is a MIMO model. A loop shaping-based $\mathcal{H}_\infty$ controller $K(s)$ is designed for each plant using the controller design procedure proposed in \cite{mcfarlane1992loop}. For demonstration purposes, the operating bandwidth $[0,\omega_0]$ rad/sec is selected arbitrarily in each example. In loop shaping, the controller $K(s)$ is designed to obtain a high loop gain of $K(s)\tilde{H}(s)$ in the operating bandwidth for good noise attenuation. Beyond the operating bandwidth, a small loop gain of $K(s)\tilde{H}(s)$ is obtained for good robust stability in the presence of uncertainties. For instance, the desired loop shape $K(s)\tilde{H}(s)\approx\frac{\omega_0}{s}$ generally provides good noise attenuation and robust stability. This loop shape ensures that $[I+K(s)\tilde{H}(s)]^{-1}K(s)\tilde{H}(s)\approx I$ within the operating bandwidth and $[I+K(s)\tilde{H}(s)]^{-1}K(s)\tilde{H}(s)\approx 0$ outside that region. Recall that the following reduction criterion ensures that $K(s)$ is also a stabilizing controller for $H(s)$
\begin{align}
\underset{\substack{\tilde{H}(s)\\\textnormal{order}=r}}{\text{min}}||[I+K(s)\tilde{H}(s)]^{-1}K(s)\tilde{H}(s)\Delta_r(s)||,
\end{align} cf. \cite{ennth}. Thus plant reduction is a frequency-limited problem with the loop shaping controller design, i.e.,
$\underset{\substack{\tilde{H}(s)\\\textnormal{order}=r}}{\text{min}}||\Delta_r(s)||$ within $[0,\omega_0]$ rad/sec. In this section, it is highlighted with numerical results that FLRHMORA provides good accuracy with respect to the relative error $\Delta_r$, similar to FLBST. The closed-loop robust stability performance of FLBT, FLBST, FLIRKA, and FLRHMORA is also compared after designing $\mathcal{H}_\infty$ controllers for the reduced-order plants constructed by these algorithms. To that end, the actual loop shape achieved with the original high-order plant, i.e., $K(s)H(s)$ and the robust stability measure $[I+K(s)\tilde{H}(s)]^{-1}K(s)\tilde{H}(s)\Delta_r(s)||_{\mathcal{H}_\infty}$ are compared.
\subsection{Experimental Setup}
All the tests are performed using MATLAB R2016a on a laptop with 16GB memory and a 2GHZ Intel i7 processor. The Lyapunov/Sylvester equations, Riccati equations, and matrix logarithms are computed using MATLAB's \textit{``lyap"}, \textit{``care"}, and \textit{``logm"} commands, respectively. FLIRKA and FLRHMORA converge within 30 iterations in all the simulations. The matrix $D$ in the three examples is zero. Therefore, it is replaced with an arbitrary chosen $D=0.0001I$ in FLBST and FLRHMORA. FLIRKA and FLRHMORA are initialized with the same random initial guess for a fair comparison. The initial guess is generated using MATLAB's \textit{``rss"} command. For demonstration purposes, the desired frequency interval is arbitrarily chosen.
\subsection{Example 1}
In this example, we examine the $348^{th}$-order SISO model of a clamped beam, sourced from the benchmark collection presented in \cite{chahlaoui2005benchmark}. For demonstration purposes, we set the bandwidth of the closed-loop system to $[0,3]$ rad/sec. Using FLBT, FLBST, FLIRKA, and FLRHMORA, we construct ROMs of orders ranging from $15$ to $40$. The corresponding relative errors $||\Delta_r(s)||_{\mathcal{H}_{2,\omega}}$ are provided in Table \ref{tab1}, highlighting the notable accuracy achieved by FLRHMORA.
\begin{table}[!h]
\centering
\caption{$\mathcal{H}_2$ norm of the relative error $\Delta_{r}(s)$}\label{tab1}
\begin{tabular}{|c|c|c|c|c|}
\hline
Order & FLBT & FLBST     & FLIRKA     & FLRHMORA \\ \hline
15     &0.3031&0.1781&0.3354&0.1048\\ \hline
20     &0.2612&0.1190&0.2895&0.0980\\ \hline
25     &0.2358&0.0812&0.2596&0.0314\\ \hline
30     &0.1988&0.0591&0.1715&0.0084\\ \hline
35     &0.1595&0.0363&0.1297&0.0059\\ \hline
40     &0.0697&0.0197&0.0395&0.0017\\ \hline
\end{tabular}
\end{table}
Figure \ref{fig1} showcases the singular values of $\Delta_{a}(s)$ for $15^{th}$-order ROMs within the desired frequency interval of $[0,3]$ rad/sec.
\begin{figure}[!h]
  \centering
  \includegraphics[width=12cm]{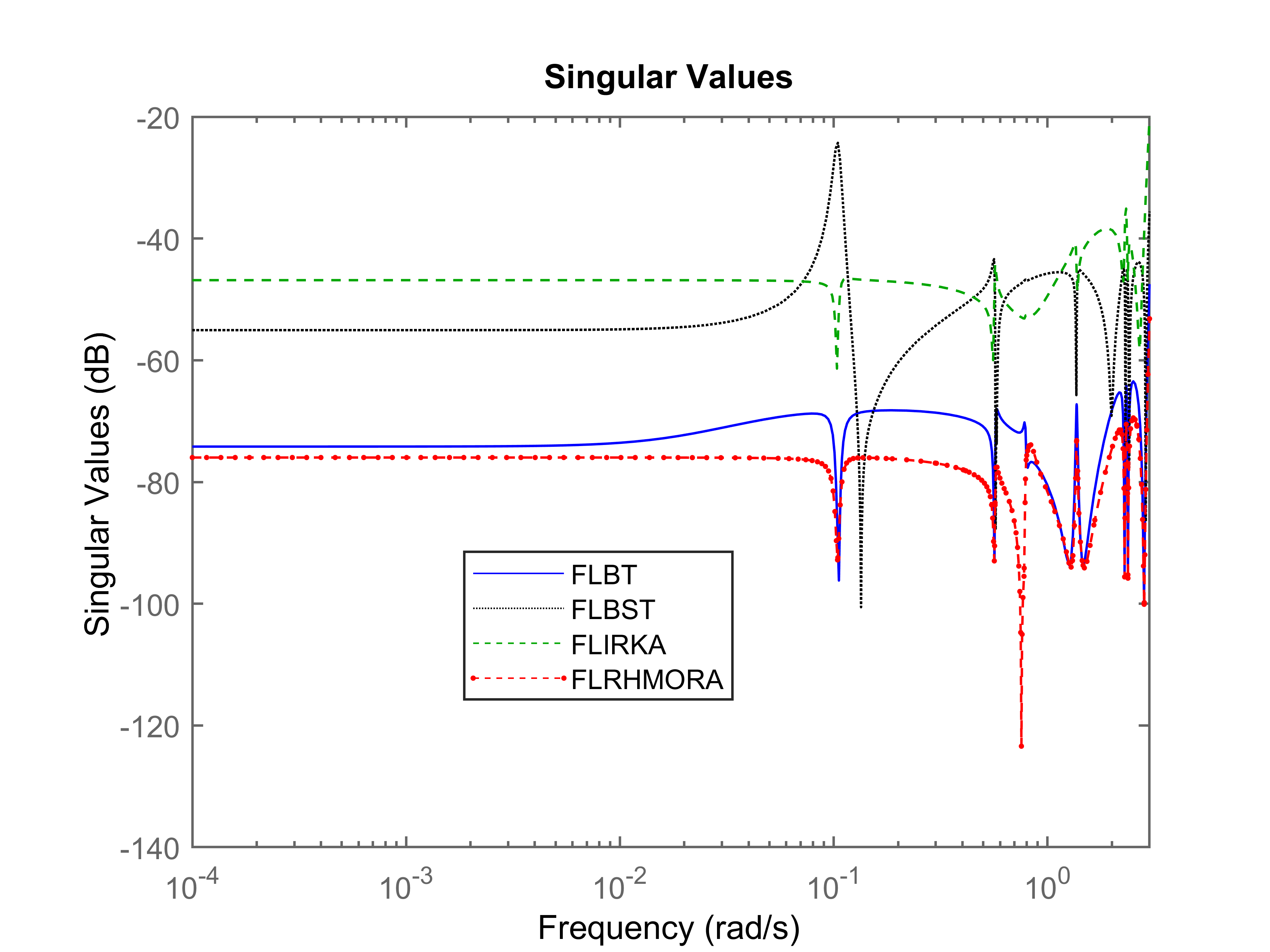}
  \caption{Singular values of $\Delta_{a}(s)$ within $[0,3]$ rad/sec}\label{fig1}
\end{figure} The accuracy achieved by FLRHMORA stands out prominently. Using the loop-shaping procedure proposed in \cite{mcfarlane1992loop}, we design an $18^{th}$-order $\mathcal{H}_\infty$ controller for each $15^{th}$-order ROM. To ensure effective noise rejection within the operating bandwidth of $[0,3]$ rad/sec, we define the desired loop shape as $3/s$. Figure \ref{fig02} visually presents the actual loop shape $K(s)H(s)$ achieved with each controller. Remarkably, the controller associated with the FLRHMORA-constructed ROM closely approximates the desired loop shape.
\begin{figure}[!h]
  \centering
  \includegraphics[width=12cm]{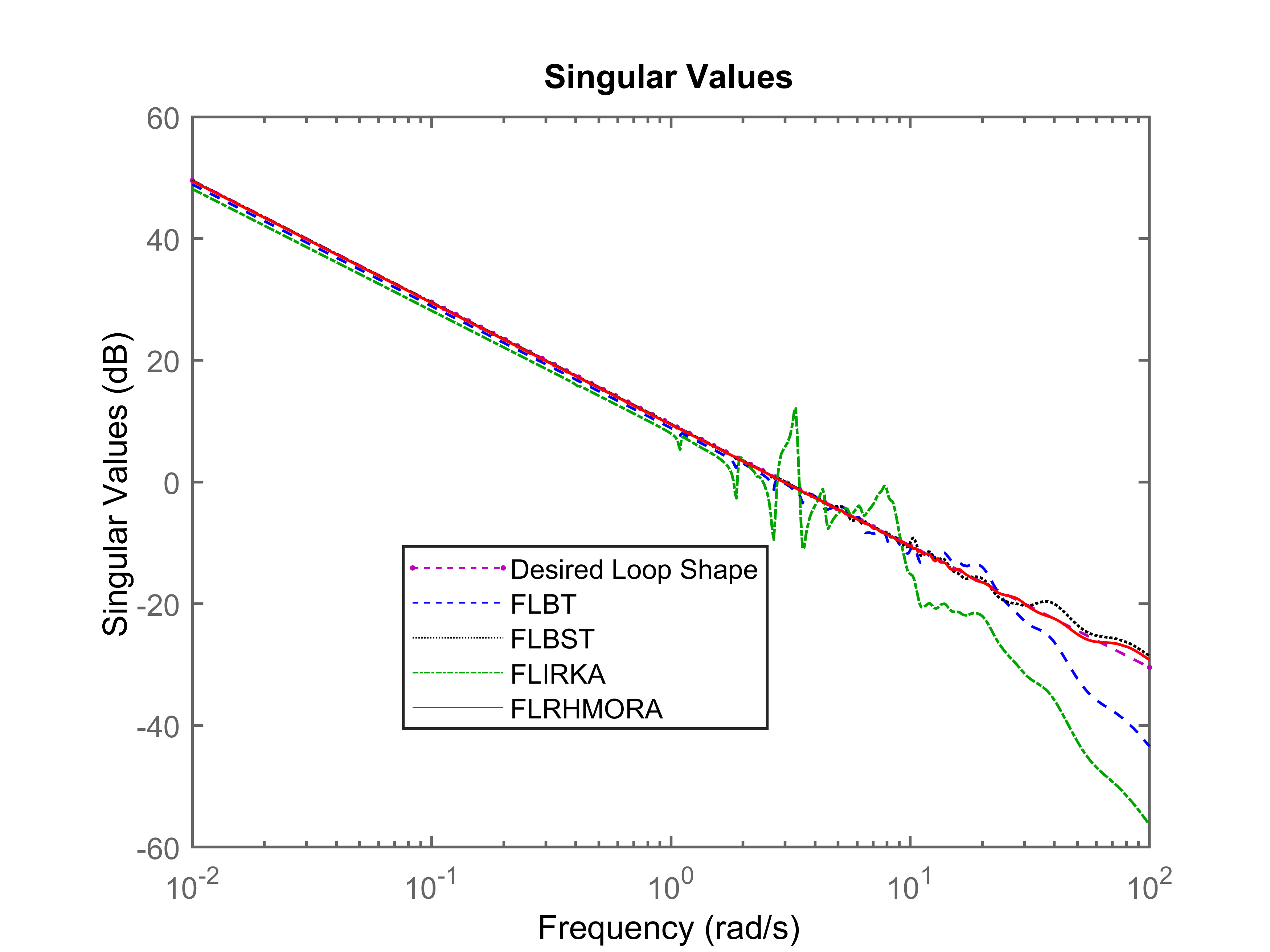}
  \caption{Actual loop shape $K(s)H(s)$}\label{fig02}
\end{figure}
Table \ref{tab02} quantifies the robustness of the controllers using the robust stability measure $||[I+K(s)\tilde{H}(s)]^{-1}K(s)\tilde{H}(s)\Delta_r(s)||_{\mathcal{H}_\infty}$, emphasizing that the FLRHMORA-constructed ROM yields the highest level of robust stability. These results affirm FLRHMORA as an effective algorithm for obtaining reduced-order controllers of superior performance for high-order plants.
\begin{table}[!h]
\centering
\caption{Robust Stability Measure}\label{tab02}
\begin{tabular}{|c|c|c|c|c|}
\hline
Method & $||[I+K(s)\tilde{H}(s)]^{-1}K(s)\tilde{H}(s)\Delta_r(s)||_{\mathcal{H}_\infty}$ \\ \hline
FLBT     & 1.8147  \\ \hline
FLBST     & 0.8127\\ \hline
FLIRKA     & 1.9058\\ \hline
FLRHMORA    & 0.3913\\ \hline
\end{tabular}
\end{table}
\subsection{Example 2}
In this example, we examine the $1006^{th}$-order SISO model of an artificial dynamical system from the benchmark collection introduced in \cite{chahlaoui2005benchmark}.  To illustrate the capabilities of the methods, we set the bandwidth of the closed-loop system as $[0,5]$ rad/sec. Utilizing FLBT, FLBST, FLIRKA, and FLRHMORA, we construct ROMs of various orders ranging from $20$ to $40$. The corresponding relative errors $||\Delta_r(s)||_{\mathcal{H}_{2,\omega}}$ are presented in Table \ref{tab2}, highlighting the commendable accuracy achieved by FLRHMORA.
\begin{table}[!h]
\centering
\caption{$\mathcal{H}_2$ norm of the relative error $\Delta_{r}(s)$}\label{tab2}
\begin{tabular}{|c|c|c|c|c|}
\hline
Order & FLBT & FLBST     & FLIRKA     & FLRHMORA \\ \hline
20     &2.5917&2.0814 &2.9097 &1.9157 \\ \hline
25     &2.3147&1.9905 &2.6278 &1.4970 \\ \hline
30     &1.9776&1.5127 &2.2546 &1.1957 \\ \hline
35     &1.2650&0.7913 &1.9957 &0.8485 \\ \hline
40     &0.9680&0.6632 &1.0964 &0.3003 \\ \hline
\end{tabular}
\end{table}
The plot in Figure \ref{fig2} showcases the singular values of $\Delta_{a}(s)$ for $20^{th}$-order ROMs within the specified frequency interval of $[0,5]$ rad/sec. The plot clearly demonstrates the accurate results achieved by FLRHMORA.
\begin{figure}[!h]
  \centering
  \includegraphics[width=12cm]{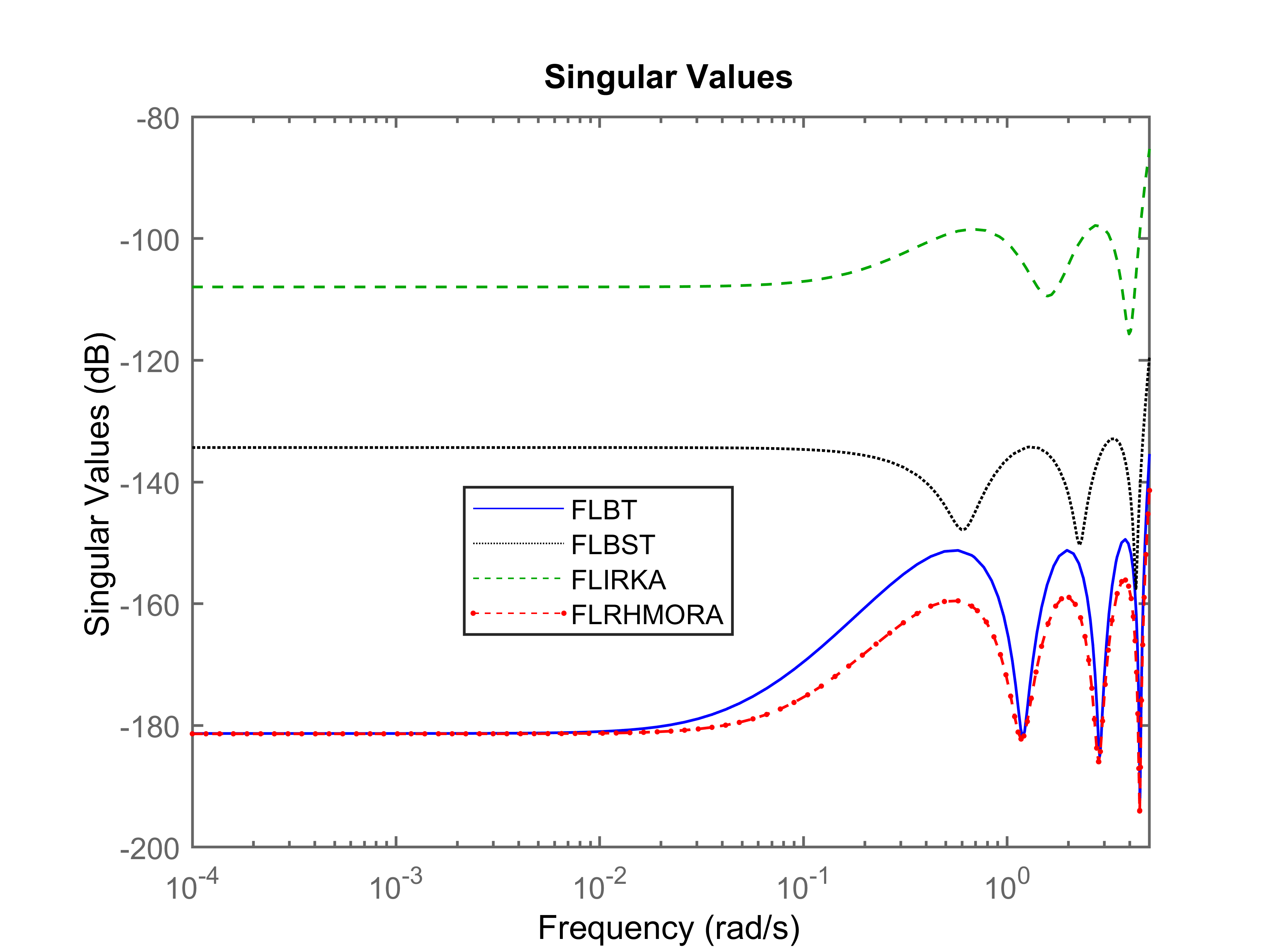}
  \caption{Singular values of $\Delta_{a}(s)$ within $[0,5]$ rad/sec}\label{fig2}
\end{figure}
For each $20^{th}$-order ROM, a $22^{th}$-order $\mathcal{H}_\infty$ controller is designed using the loop-shaping procedure outlined in \cite{mcfarlane1992loop}. In order to achieve efficient noise rejection within the frequency range of $[0,5]$ rad/sec, the desired loop shape is defined as $5/s$. Figure \ref{fig03} depicts the resulting loop shapes obtained with each controller, demonstrating that the controller derived from the ROM constructed by FLRHMORA closely matches the desired loop shape.
\begin{figure}[!h]
  \centering
  \includegraphics[width=12cm]{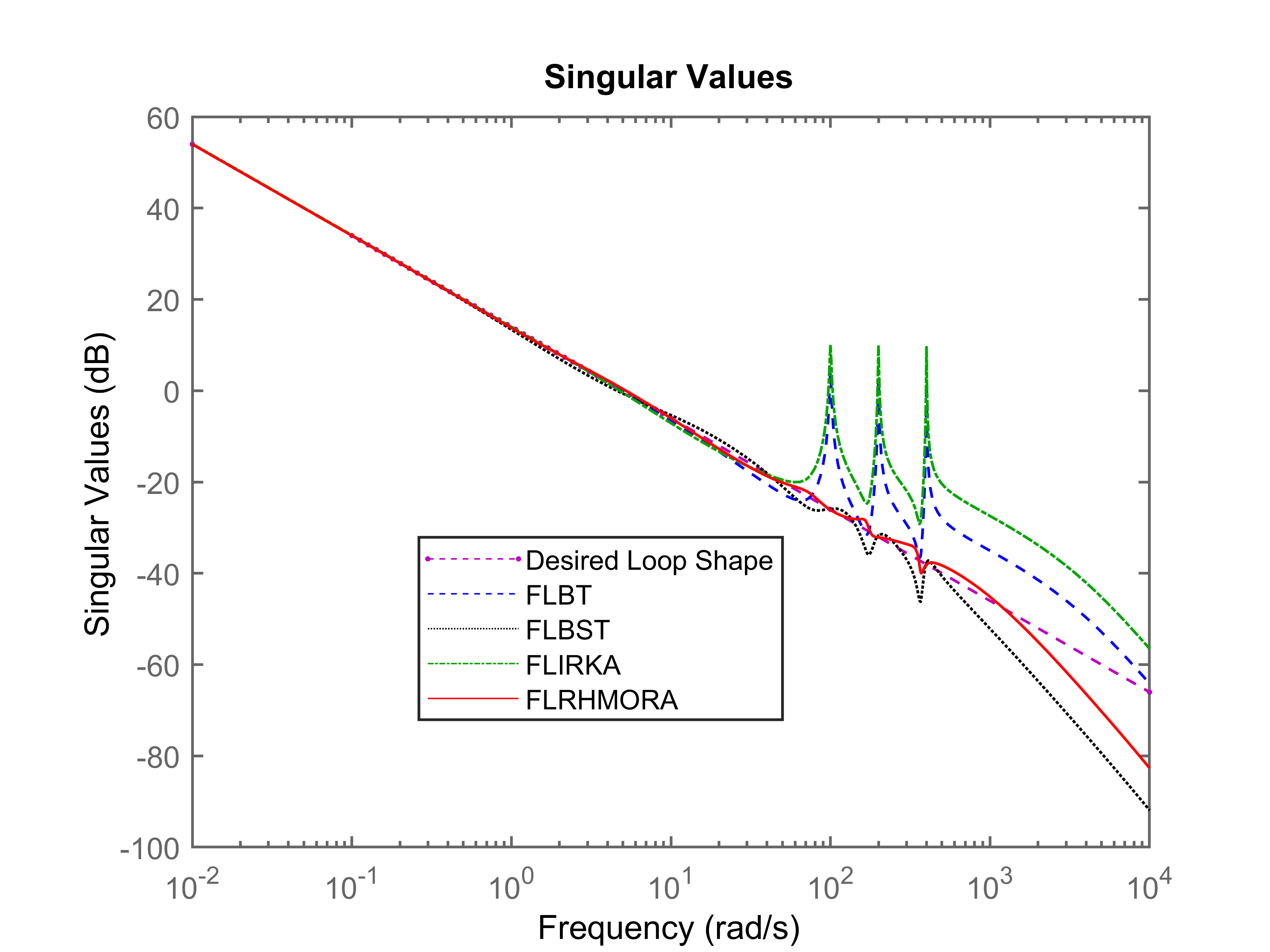}
  \caption{Actual loop shape $K(s)H(s)$}\label{fig03}
\end{figure}
To assess the robustness of the controllers, Table \ref{tab03} provides the robust stability measure $||[I+K(s)\tilde{H}(s)]^{-1}K(s)\tilde{H}(s)\Delta_r(s)||_{\mathcal{H}_\infty}$, which clearly indicates that the controller associated with FLRHMORA exhibits the highest level of robust stability.
\begin{table}[!h]
\centering
\caption{Robust Stability Measure}\label{tab03}
\begin{tabular}{|c|c|c|c|c|}
\hline
Method & $||[I+K(s)\tilde{H}(s)]^{-1}K(s)\tilde{H}(s)\Delta_r(s)||_{\mathcal{H}_\infty}$ \\ \hline
FLBT     & 2.2785  \\ \hline
FLBST     & 1.1548\\ \hline
FLIRKA     & 3.9575\\ \hline
FLRHMORA    & 0.8485\\ \hline
\end{tabular}
\end{table}
\subsection{Example 3}
For this example, we examine the Brazil interconnected power system, a MIMO model with four inputs and four outputs, which has an order of $3077$ as reported in \cite{rommes2009computing}. To demonstrate the performance of the methods, the bandwidth of the closed-loop system is set to $[0,7]$ rad/sec.  Employing FLBT, FLBST, FLIRKA, and FLRHMORA, ROMs of order $30$ are constructed. A comparison of the simulation time needed to compute the ROMs using each algorithm is presented in Table \ref{tabst}.
\begin{table}[!h]
\centering
\caption{Computational Efficiency}\label{tabst}
\begin{tabular}{|c|c|c|c|c|}
\hline
Method & Simulation Time (sec)\\ \hline
FLBT     & 216.1430  \\ \hline
FLBST     & 1152.0357\\ \hline
FLIRKA     & 31.6787\\ \hline
FLRHMORA    & 49.1712\\ \hline
\end{tabular}
\end{table}
The computational comparison reveals that FLRHMORA is slightly more computationally demanding than FLIRKA, yet significantly more efficient than FLBT and FLBST.  For each $30^{th}$-order ROM, a $33^{th}$-order $\mathcal{H}_\infty$ controller is designed using the loop-shaping procedure described in \cite{mcfarlane1992loop}. To achieve effective noise rejection within the operating bandwidth of $[0,7]$ rad/sec, the desired loop shape is specified as $7/s$. Figure \ref{fig04} showcases the achieved loop shape for each controller.
\begin{figure}[!h]
  \centering
  \includegraphics[width=12cm]{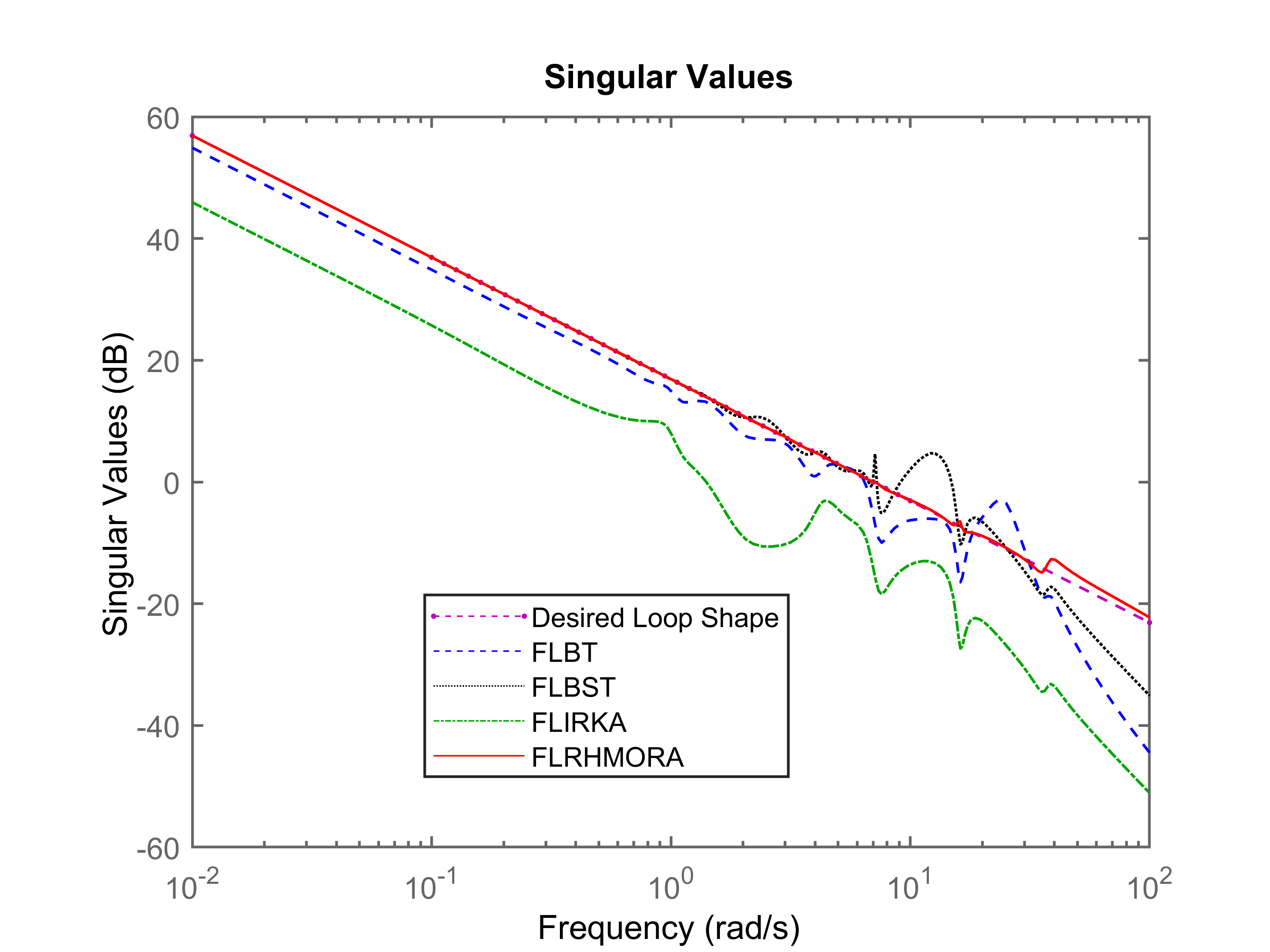}
  \caption{Actual loop shape $K(s)H(s)$}\label{fig04}
\end{figure}
The controller designed for the ROM constructed by FLRHMORA exhibits a loop shape that closely resembles the desired loop shape. In order to assess the controllers' robustness, the robust stability measure $||[I+K(s)\tilde{H}(s)]^{-1}K(s)\tilde{H}(s)\Delta_r(s)||_{\mathcal{H}_\infty}$ is presented in Table \ref{tab04}. The tabulated results highlight that the controller associated with the ROM constructed by FLRHMORA exhibits the best robust stability among the considered controllers.
\begin{table}[!h]
\centering
\caption{Robust Stability Measure}\label{tab04}
\begin{tabular}{|c|c|c|c|c|}
\hline
Method & $||[I+K(s)\tilde{H}(s)]^{-1}K(s)\tilde{H}(s)\Delta_r(s)||_{\mathcal{H}_\infty}$ \\ \hline
FLBT     & 3.2327  \\ \hline
FLBST     & 1.2569\\ \hline
FLIRKA     & 3.8184\\ \hline
FLRHMORA    & 0.9742\\ \hline
\end{tabular}
\end{table}
\section{Conclusion}
This research focuses on addressing the problem of frequency-limited relative error $\mathcal{H}_2$ MOR and presents the necessary conditions for achieving a stable minimum phase local optimum. Building upon these conditions, a generalized oblique projection algorithm is proposed, eliminating the need for the ROM to maintain stable minimum phase characteristics in every iteration. In contrast to FLBST, the proposed algorithm relies on solving small-scale Lyapunov and Riccati equations instead of computationally demanding large-scale ones, resulting in low computational costs. The oblique projection reduction matrices can be obtained by solving ``sparse-dense" Sylvester equations, for which effective numerical methods are available. The algorithm's performance is demonstrated through the design of reduced-order loop shaping controllers for benchmark high-order plants, with numerical results confirming its accuracy and robust stability when the reduced-order controllers are integrated in a closed-loop with the original high-order plant.
\section*{Acknowledgment}
This work is supported by the National Natural Science Foundation of China under Grants No. 61873336 and 62273059, the International Corporation Project of Shanghai Science and Technology Commission under Grant 21190780300, in part by The National Key Research and Development Program No. 2020YFB1708200 , and in part by the High-end foreign expert programs No. QN2022013008L, G2021013008L, and G2022013050L granted by the State Administration of Foreign Experts Affairs (SAFEA).
\section*{Competing Interest Statement}
The authors declare no conflict of interest.

\end{document}